\documentclass[journal]{IEEEtran}
\usepackage{amsmath,amsfonts}
\usepackage{amsmath}
\usepackage{algorithmic}
\usepackage{algorithm}
\usepackage{array}
\usepackage[font=small]{caption}
\usepackage{textcomp}
\usepackage{stfloats}
\usepackage{url}
\usepackage{verbatim}
\usepackage{graphicx}
\usepackage{cite}
\usepackage{tikz}
\usepackage{tikzit}
\usepackage{amsthm}
\usepackage{subcaption}
\usepackage[normalem]{ulem}
\usepackage{pdfpages}

\usepackage{flushend}
\usepackage{amssymb}

\newtheorem{remark}{Remark}
\newtheorem{corollary}{Corollary}
\usepackage{mathtools}
\newtheorem{proposition}{Proposition}
\usepackage{breqn}
\usepackage[multiple]{footmisc}
\newcommand{\widesim}[2][1.5]{
  \mathrel{\overset{#2}{\scalebox{#1}[1]{$\sim$}}}
}

\tikzstyle{new edge style 0}=[-]
\tikzstyle{dashes}=[-, fill=none, thick, dashed]
\tikzstyle{new edge style 1}=[<->]

\newtheorem{theorem}{Theorem}
\newtheorem{lemma}{Lemma}
\DeclareMathOperator*{\argmax}{arg\,max}
\allowdisplaybreaks

\begin{document}

\title{Performance Analysis of Intelligent Reflecting Surface Assisted Opportunistic Communications}

\author{L. Yashvanth,~\IEEEmembership{Student Member,~IEEE}, Chandra R. Murthy,~\IEEEmembership{Senior Member,~IEEE}
\thanks{The authors are with the Department of Electrical Communication Engineering, Indian Institute of Science, Bangalore, India 560 012. (E-mails: \{yashvanthl, cmurthy\}@iisc.ac.in). This work was supported in part by the Qualcomm Innovation Fellowship, 2021.}
}



\maketitle

\begin{abstract}
Intelligent reflecting surfaces (IRSs) are a promising technology for enhancing coverage and spectral efficiency, both in the sub-6 GHz and the millimeter wave (mmWave) bands. Existing approaches to leverage the benefits of IRS involve the use of a resource-intensive channel estimation step followed by a computationally expensive algorithm to optimize the reflection coefficients at the IRS. In this work, focusing on the sub-6 GHz band of communications, we present and analyze several alternative schemes, where the phase configuration of the IRS is randomized and multi-user diversity is exploited to opportunistically select the best user at each point in time for data transmission. 
We show that the throughput of an IRS assisted opportunistic communication (OC) system asymptotically converges to the optimal beamforming-based throughput under fair allocation of resources, as the number of users gets large. We also introduce schemes that enhance the rate of convergence of the OC rate to the beamforming rate with the number of users.  For all the proposed schemes, we derive the scaling law of the throughput in terms of the system parameters, as the number of users gets large. 
Following this, we extend the setup to wideband channels via an orthogonal frequency division multiplexing (OFDM) system and discuss two OC schemes in an IRS assisted setting that clearly elucidate the superior performance that IRS aided OC systems can offer over conventional systems, at very low implementation cost and complexity. \end{abstract}


\section{Introduction}
Intelligent Reflecting Surfaces (IRSs) have become a topic of active research for enhancing the performance of next generation wireless communication systems \textcolor{black}{both in the sub-6 GHz and in the millimeter wave (mmWave) bands.}  
An IRS consists of passive elements made out of meta-materials that can be tuned to offer a wide range of load impedances  using a PIN diode. 
\textcolor{black}{Using this,} 
each element of the IRS can be tuned to have a different reflection coefficient, and thereby enable the IRS to  
reflect the incoming signals in any desired direction~\cite{Renzo_EURASIP_2019,Basar_IA_2019,Wu_ICOMM_2020,Emil_arxiv_2021}. 
\textcolor{black}{However, realizing these benefits entails high overheads in terms of resource-intensive channel estimation followed by solving a computationally heavy optimization problem to determine the phase configuration at the IRS. In this work, we consider an alternative approach, 
where the phase configuration of the IRS is set randomly in each slot but yet extracts the benefits from the IRS in terms of the enhancement of the system throughput. This approach only requires a short training signal for estimating the received signal power at the users, followed by feedback-based selection of the best user in each slot for subsequent data transmission. Multi-user diversity ensures that at least one user will see a good channel in the randomly chosen phase configuration~\cite{Viswanath_TIT_2002}}.

\textcolor{black}{Despite its short history, significant work has gone into the design and optimization of IRS-aided communication systems. Here, we briefly summarize the existing literature, in order to place the contributions of this paper in context. In \cite{Wang_IVT_2020}, the authors show that an IRS can create a virtual line-of-sight (LoS) path between the base station (BS) and user, leading to improved coverage and SNR in mmWave systems. 
In \cite{Wu_GLOBECOM_2018}, it is shown that the received SNR increases quadratically with the number of IRS elements, provided the phase configuration of the IRS is optimized to ensure coherent combining of the signal at the receiver location. 
Also, since the IRS is passive in nature, it boosts the spectral efficiency without compromising on  the energy efficiency~\cite{Huang_TWC_2019}.
In \cite{Guo_GLOBECOM_2019}, the authors propose joint active and passive beamforming algorithms at the BS and IRS, respectively, to  maximize the weighted sum rate of an IRS-aided system.
IRS phase optimization in the context of multiple-input multiple-output (MIMO) and orthogonal frequency division multiplexing~(OFDM) systems have been studied in \cite{Yang_TCOM_2020, Rui_TWC_2020, Lin_TWC_early,Li_TCOM_2021,Li_WCNC_2020}, and the list of potential applications of IRS continues to grow~\cite{Wu_ICOMM_2020,Wu_TCOM_2021,Zhao_arxiv,Chen_CC_2021}.}

All of the above mentioned works  describe and solve a complex phase optimization problem, which are computationally intensive and difficult to implement in practical real-time systems. Further, this optimization becomes even more complex in the context of OFDM systems, as it requires one to optimize the IRS jointly across all the OFDM subcarriers.
 More importantly, these phase optimization algorithms work on the premise of the availability of accurate channel state information (CSI) of the links between the BS and the user through every IRS element. 
Elegant methods for channel estimation in IRS aided systems are described in \cite{Mishra_ICASSP_2019,Nadeem_JCS_2020}, but in all these schemes, the channel estimation overhead  scales linearly with the number of IRS elements. The time, energy and resource utilization for channel estimation can quickly erase much of the benefits offered by the IRS.  
One approach to mitigate this loss is to exploit structure in the channel model to estimate the channel with lower overhead~\cite{Wei_1_ICL_2021,Wei_2_ICL_2021,Lin_arxiv}, but these approaches trade-off the reduction in overhead with more complex channel estimation algorithms, thereby substantially increasing the computational cost.
   In addition, the complexity of the overall algorithm increases with the resolution with which phase shifts are configured~\cite{Wu_TCOM_2020}.
Thus, regardless of how they are implemented, the use of IRS incurs significant computational and training overheads in order to fully reap their professed benefits. 
     Furthermore, since the IRS is passive, these optimization algorithms have to run at the BS, and a dedicated control link from the BS to the IRS is needed to communicate the phase configuration information to the IRS.  As the number of IRS elements increases, this becomes an additional bottleneck, as the control link overhead also scales with the number of IRS elements \cite{Wu_ICOMM_2020}. 

\textcolor{black}{In the context of the above, it is pertinent to explore whether one can circumvent the channel estimation and phase optimization overheads and still obtain most of the benefits of an IRS assisted system. An interesting approach in this context is to configure the IRS with random phases and make the communications opportunistic in nature. 
In opportunistic communications (OC), at every point in time, we serve the user who witness the best instantaneous channel condition. When there are a large number of users in the system,\footnote{\textcolor{black}{In sub-6 GHz band communication system, the consideration of large number of users is realistic~\cite{Yang_SMJ_2022,Nadeem_TWC_2021,Nadeem_WCL_2021,Viswanath_TIT_2002}, especially in the context of massive machine-type communications (mMTC)~\cite{Mahmood_EURASIP_2021}.}} with high probability, deep fade events are avoided at any given user, enhancing the average system throughput without incurring the three fold overheads mentioned above~\cite{Viswanath_TIT_2002,Hassibi_TIT_2005, Asadi_CST_2013}. 
In particular, opportunistic scheduling is pertinent when the goal is to  maximize the system average throughput, i.e., the average sum-rate across the users over a long time horizon. Such an approach is suitable in delay-tolerant networks, where users can afford to wait  before being scheduled for transmission.} 
  
\textcolor{black}{Initial work along these lines was reported in \cite{Nadeem_WCL_2021}, where the phase angles of the reflection coefficients at the IRS elements are drawn uniformly and independently from the interval $[0, 2\pi)$.  As we show in the sequel, a drawback of this approach is that the number of users needed to achieve a performance comparable to coherent beamforming increases exponentially with the number of IRS elements, making it unattractive for practical implementation. Moreover, the \textcolor{black}{average} effective SNR scales linearly, not quadratically, in the number of IRS elements (as achieved by coherent beamforming.) 
In this work, we develop novel, alternative schemes that overcome these drawbacks. Although we focus on communications over the sub 6-GHz bands (FR-1 band in the 5G NR specifications~\cite{Pai_5G_spectrum}) assisted by an IRS, we also briefly discuss how one of the} \textcolor{black}{proposed} \textcolor{black}{schemes is relevant in mmWave bands also. For all the schemes, we analyze the  system throughput as a function of the number of users. We show that, by exploiting the structure in the channel, we can significantly improve the convergence rate of opportunistic throughput to the beamforming throughput and also achieve the quadratic scaling of the SNR with the number of IRS elements. This, in turn, allows us to achieve near-optimal beamforming performance and also obtain an additional gain from opportunistic user selection, without requiring a very large number of users in the system.} 

\textcolor{black}{
The specific contributions of our work are as follows:
\begin{itemize}
\item We analyze the throughput of several proposed IRS-assisted OC schemes for independent and identically distributed narrowband wireless channels. 
We exploit the fast switching time of IRS phase configurations to obtain additional reflection diversity from the IRS, and show that this helps to reduce the number of users required to obtain near-optimal throughput. We also analytically characterize the throughput achievable by this scheme. 
(See Theorem~\ref{thm:mulitple_pilot_rate} and Sec. \ref{sec:sel_div_scheme_nb}.)
\item Next, we consider directional  channels in the IRS aided system and design channel-aware randomly configured OC schemes that converge to the coherent beamforming rate without requiring the users \textcolor{black}{to scale exponentially with the number of IRS elements.}
In  Theorem~\ref{prop:rate_irs_channel_aware}, we show that not only does this scheme achieve the quadratic scaling of the SNR with the number of IRS elements, its throughput can even surpass that of the scheme that involves IRS optimization techniques, due to the lack of multi-user diversity gain in the latter.  We also discuss how this scheme can be applied to mmWave channels which also bear a similar structure. (See Sec.~\ref{subsec:steering_model_scheme}.)
\item We extend  the OC schemes to IRS aided systems with wideband wireless channels. Specifically, we consider an OFDM system and discuss two OC schemes, namely, single user OFDM where we schedule a single user across all subcarriers, and orthogonal frequency division multiple access (OFDMA) where multiple users are potentially scheduled across the subcarriers. We derive the sum throughput scaling laws in Theorems~\ref{thm:su_ofdm_rate_scale} and~\ref{thm:ofdma_rate_law} for the two schemes, and provide interesting insights about these~systems. (See Sec.~\ref{sec.opp_schemes_wideband}.)
\end{itemize}}
\vspace{-0.1cm}
\indent The results (in Sec.~\ref{sec.results}) show that the presence of an IRS can significantly enhance the throughput of conventional BS-assisted OC schemes~\cite{Viswanath_TIT_2002}. Specifically, the throughput of IRS aided OC grows with the number of IRS elements $N$, whereas such growth is not possible in BS  assisted OC as the number of antennas at the BS is increased. This is due to the power constraint at transmitter, which eventually limits the maximum achievable throughput. Secondly, the numerical results elucidate the significant reduction in the number of users to achieve the optimal throughput compared to existing schemes such as in~\cite{Nadeem_WCL_2021}. For example, in an $8$-element IRS system, the approach in~\cite{Nadeem_WCL_2021} has a gap of $175\%$ from the \textcolor{black}{optimal rate}; this gap reduces to $60\%$ by using the proposed reflection diversity enhanced scheme.
Further, the offset from \textcolor{black}{the coherent beamforming throughput} reduces to $11\%$ in the proposed channel model aware IRS assisted OC scheme. Also, the 
the channel model aware OC scheme is within a small offset ($18\%$) with a modest number of users ($\approx 50$), even when the number of IRS elements is as large as $1024$. Thus, IRS aided OC is a promising approach for exploiting the benefits of IRS-aided systems without incurring the cost of training, phase angle optimization, and communication to the IRS.


\emph{Notation:} 
$[N]$ stands for the set of natural numbers from $1$ to $N$; 
$|\cdot|,\angle\cdot$ stand for the magnitude and phase of a complex number (vector); 
$\|\cdot\|_p$ denotes the $\ell_p$ vector norm;
 $\mathcal{CN}(\boldsymbol{\mu},\mathbf{\Sigma})$ denotes a circularly symmetric complex Gaussian random vector with mean $\boldsymbol{\mu}$ and covariance matrix $\mathbf{\Sigma}$, $\mathcal{U}(\phi_0,\phi_1)$ denotes a uniformly distributed random variable with support $[\phi_0,\phi_1]$, $\exp(\lambda)$ denotes an exponentially distributed random variable with parameter $\lambda$; 
$Pr(\cdot)$ refers to the probability measure, and 
$\mathcal{O}(\cdot)$ is the Landau's Big-O notation.

 
\section{Preliminaries}\label{sec.prelims}
Opportunistic communication schemes exploit the 
variation of the fading channels across users in order to improve the throughput of a multi-user system.
For example, in max-rate based scheduling~\cite{Asadi_CST_2013}, 
the BS sends a common pilot signal to all the users in the system, and the users measure the received SNR. The BS then collects feedback from the user who witnesses the highest SNR,\footnote{We note that various timer-based and splitting based schemes can be used to identify the best user with low overhead \textcolor{black}{\cite{suresh_NCC_2010,Shah_TCOM_2010}}.} and schedules data to that user in the rest of the slot. 
That is, in a $K$ user system, the BS serves user $k^*$ at time $t$, where  
    \begin{equation*}
        k^* = \argmax\limits_{k \in [K]}\text{ }|h_k(t)|^2,
    \end{equation*}
with $h_k(t)$ denoting the channel seen by user $k$ at time $t$. 
However, this scheme is unfair to users located far away from the BS due to their higher path loss. 
An alternative is to consider the proportional fair (PF) scheduling scheme, which provides a  trade-off between fairness and system throughput~\cite{Viswanath_TIT_2002}. 
The PF scheduler serves user $k^*$ at time $t$ such that
       \begin{equation}\label{eq:pf_user_sel}
           k^* = \argmax_{k \in [K]} {\frac {R_k(t)}{T_k(t)}},
       \end{equation} 
       where $R_k(t) = \log_2\left(1+\frac{P|h_k(t)|^2}{\sigma^2}\right)$ 
       is the achievable rate\footnote{In this paper, we use rate and throughput interchangeably.}  of user $k$ in the current time slot, $P$ is the transmit power at the BS, $\sigma^2$ is the noise variance at the user, and $T_k(t)$ captures the long-term average throughput of user $k$. It is updated as
\begin{equation}\label{eq:tk_update}
            T_k(t+1) = \begin{cases}\left(1- { {\frac {1}{\tau}}}\right)T_k(t) + { {\frac {1}{\tau}}}\, R_k(t),\quad & k=k^*, \\ \left(1-{ {\frac{1}{\tau}}}\right)T_k(t), \quad & k \neq k^*.\end{cases}
        \end{equation}
    
Here, the variable $\tau$ represents the length of a window that captures the tolerable latency of the application and dictates the trade-off between fairness and throughput. Going forward,
we will refer to the term $\frac{R_k(t)}{T_k(t)}$ as the PF metric.    

A simple illustration of exploiting multi-user diversity through opportunistic scheduling based on the PF scheduler is shown in Fig. \ref{fig:pf_scheduler} for various values of $\tau$. We consider a time-varying channel across time slots modeled as a Gauss-Markov process, i.e., the channel at user $k$ varies  as 
\begin{equation}
        h_k(t) = \alpha h_k(t-1) + \sqrt{1-\alpha^2}v_k(t),
    \end{equation}
where $\alpha$ dictates the correlation of channel coefficients across time slots and $v_k(t) \sim \mathcal{CN}(0,1)$ is an innovation process. 
\begin{figure}[t]
        \centering
        \includegraphics[scale=0.58]{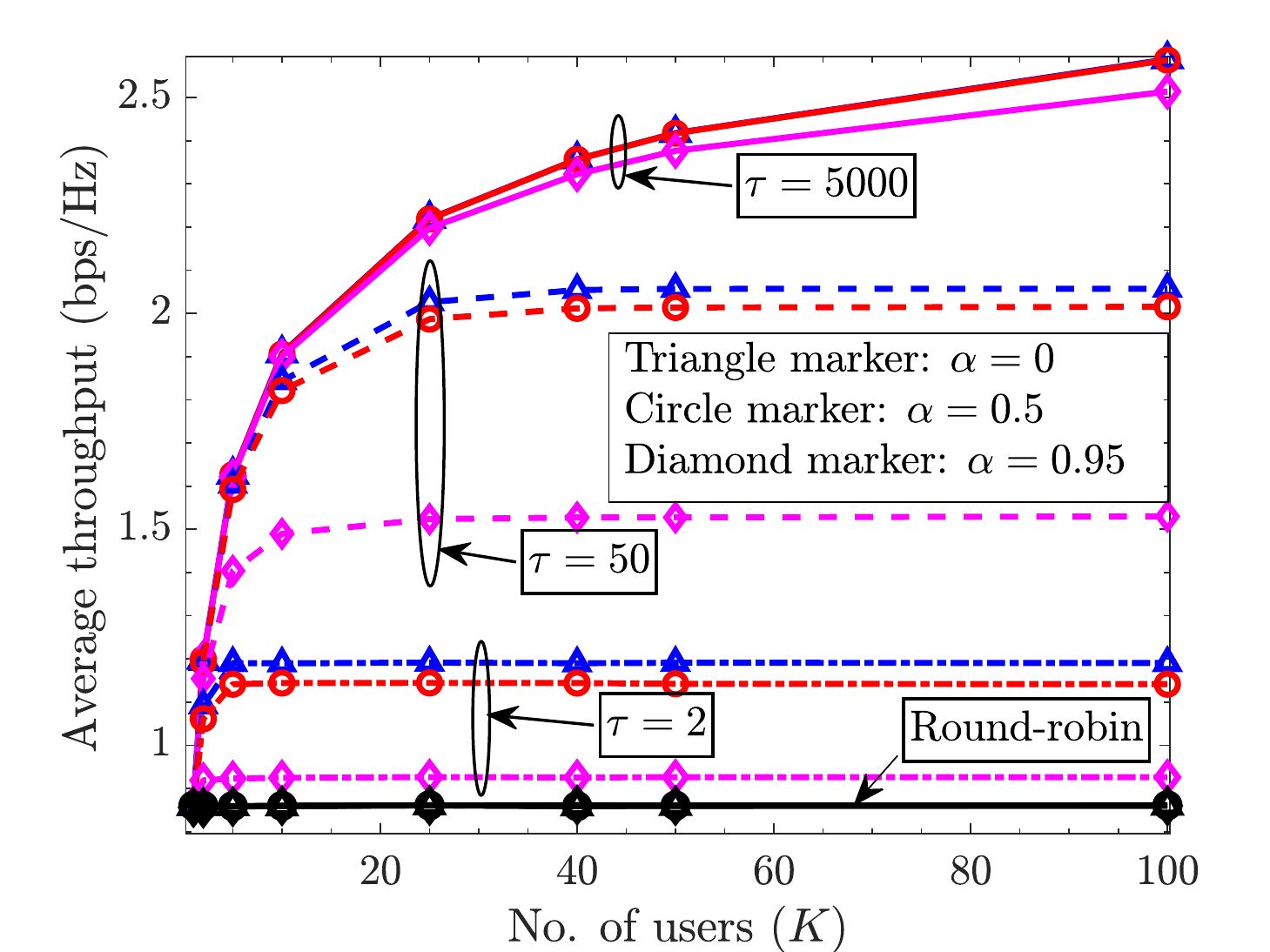}
        \caption{Throughput achieved via PF scheduling.}
        \label{fig:pf_scheduler}
    \end{figure}
    
Further, the performance of the OC scheme is compared with a non-opportunistic scheduling scheme, namely, round-robin time division multiple access (TDMA). 

As can be seen from the figure, the average throughput of the system at higher $\tau$ grows substantially with the number of users. This marks the rate-constrained regime, where the instantaneous rate is the primary factor determining which user gets scheduled (the system still offers fairness, but only over very large time-scales). Thus, for large $\tau$, the PF scheduler is approximately the same as a max-rate scheduler, and the throughput achieved by the PF scheduler approaches that achieved by the max-rate scheduler as $\tau$ goes to infinity. On the other hand, at lower $\tau$, the average throughput has negligible improvement after a few users, indicating that this is a fairness-constrained regime where users are selected in a nearly round-robin fashion to ensure short-term fairness. 
On the other hand, for a given choice of $\tau$, opportunistic scheduling performs better for lower values of $\alpha$ (representing a fast-fading scenario) compared to higher values (representing slow-fading scenarios). This is because, in fast-fading environments, the rate of channel fluctuations are enhanced, which improves the performance of opportunistic scheduling schemes. 

We note that such an enhancement of the rate of channel fluctuations can be obtained by choosing different, random phase configurations at an IRS. This has the additional advantage that the fluctuations induced by the IRS increases with the number of elements. This motivates us to 
take a fresh look at the opportunistic scheduling schemes in IRS-assisted scenarios. 
Specifically, we analyze the achievable throughput in IRS-assisted OC schemes. We show that one can obtain the benefits of using an IRS over conventional systems with low overhead and complexity by obviating the need for CSI acquisition, IRS phase optimization and feedback between the BS and the IRS.
\section{Single IRS Assisted Opportunistic User Scheduling for Narrowband Channels}\label{sec.opp_scheme_narrowband}
In this section, we present three OC schemes in a single IRS assisted  setting, for narrowband channels. We consider a single cell containing a BS equipped with one antenna serving $K$ single antenna users. An IRS equipped with $N$ reflecting elements is deployed at a suitable location in the radio propagation environment, as shown in Fig.~\ref{fig:IRS_model}. 
\subsection{IRS-Enhanced Multi-user Diversity}\label{sec:basic_scheme_nb}
\subsubsection{Channel Model}\label{sec:basic_ch_model}
 The signal transmitted by the BS reaches each user via a direct path as well as via the IRS. Thus, the effective downlink channel seen by user $k$ (at time slot $t$), denoted by $h_k$ (we omit the dependence on $t$ for notational brevity), is given by
\begin{equation}\label{eq:basic_channel}
	h_k = \sqrt{\beta_{r,k}}\mathbf{h}_{2,k}^H\mathbf{\Theta h}_1 +  \sqrt{\beta_{d,k}}h_{d,k},
\end{equation}
where $\mathbf{h}_{2,k}$ and  $\mathbf{h}_1 \in \mathbb{C}^{N\times 1}$  represent the channels between the IRS and user $k$, and between the BS and IRS, respectively, and $h_{d,k}$ denotes the direct  non-IRS channel between the BS and user $k$. 
 We model $\mathbf{h}_1 \sim \mathcal{CN}(\mathbf{0},\mathbf{I})$, $\mathbf{h}_{2,k} \widesim[2]{\text{i.i.d.}} \mathcal{CN}(\mathbf{0},\mathbf{I})$ and $h_{d,k} \widesim[2]{\text{i.i.d.}} \mathcal{CN}(0,1)$ across all users.
 Further, $\beta_{r,k}$ and $\beta_{d,k}$ represent the path loss between the BS and user $k$ through the IRS and direct paths, respectively. The diagonal matrix $\mathbf{\Theta} \in \mathbb{C}^{N \times N}$ contains the reflection coefficients programmed at the IRS, with each diagonal element being of the form $e^{j\theta_i}$, where $\theta_i \in [0, 2\pi)$ is the phase angle of the reflection coefficient at the $i$th IRS element. 
The signal received at every user is corrupted by AWGN $\widesim[2]{\text{i.i.d.}} \mathcal{CN}(0, \sigma^2)$. 
  
  \begin{figure}[t]
  \centering
  \includegraphics[scale=0.17]{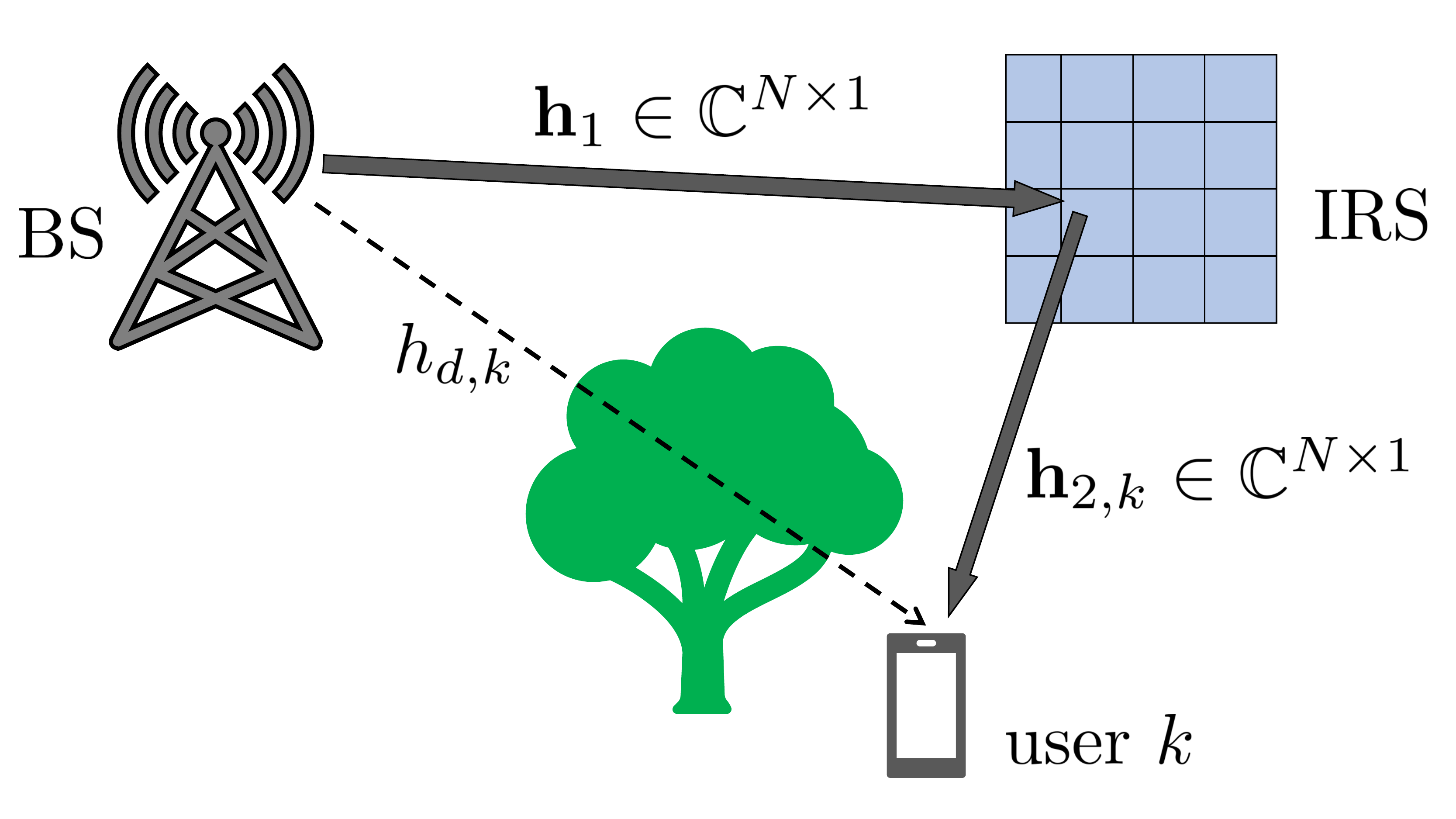}
  \caption{A single IRS assisted wireless system.}
  \label{fig:IRS_model}
  \end{figure}

\subsubsection{Scheme for IRS-Enhanced Multi-user Diversity} \label{sec:IRS_basic_scheme}
\textcolor{black}{In every time slot, the IRS sets a random phase configuration. Consequently, the effective channels seen by the users change in every slot. The BS transmits a pilot signal in the downlink at the start of the slot. The users measure the SNR from the  pilot signal, and compute their respective PF metrics. The user with the highest PF metric feeds back its identity to the BS which schedules that user for data transmission for the rest of~slot.} \\
\indent \textcolor{black}{ \emph{Feedback Mechanism:}  We consider timer or splitting based schemes~\cite{suresh_NCC_2010,Shah_TCOM_2010} for identifying the best user at the BS. These are low overhead \emph{distributed} user selection schemes, where only the best user transmits its identity to the BS, based on their channel-utility metric (e.g., using a timer that expires after a time interval that is inversely proportional to the SNR or the PF metric). It is known that, with these schemes, the BS can identify the best user within 2 or 3 (mini-)slots on average even as $K \rightarrow \infty$ \cite{suresh_NCC_2010,Shah_TCOM_2010}. Since the average overhead of a timer based feedback scheme is small compared to the time slot duration, we ignore its effect in this paper.}

In the above scheme, as the number of users in the system grows, the randomly chosen IRS configuration is likely to be close to the beamforming (BF) configuration for at least one of the users in the system, similar to~\cite{Viswanath_TIT_2002}. 
Note that, in this scheme,  
there is no communication from the BS to IRS,  making it attractive from an implementation perspective. Thus, the benefits of an optimized IRS can be readily obtained without requiring careful optimization of the IRS, provided there are a large number of users in the system \textcolor{black}{and the multi-user diversity gain is exploited.} \footnote{\textcolor{black}{Note that considering a large number of users  is realistic in 5G and beyond communications, where one of the key use-cases is to support massive machine type communications (mMTC)~\cite{Mahmood_EURASIP_2021}.}} We begin our discussion with the following lemma on the performance of an IRS that adopts a beamforming configuration to a given user, which will serve as a benchmark for evaluating the OC based schemes.
\begin{lemma}[\cite{Nadeem_WCL_2021}]\label{thm:basic_optimal_rate}
The rate achieved by user $k$ in an IRS aided system under the   \emph{beamforming configuration} is $R_{k}^{BF}
=$
\vspace{-0.2cm}
\textcolor{black}{\begin{equation}\label{eq:basic_optimal_rate} 
\log _{2}\!\left( \! 1\! +\! \frac {P}{\sigma ^{2}} \left|\sqrt {\beta _{r,k}} \sum _{n=1}^{N} |{h}_{1,n}{h_{2,k,n}}| 
   +  \sqrt {\beta _{d,k}} |h_{d,k}|\right|^{2}\right),
\end{equation}}
\vspace{-0.2cm}
with the beamforming configuration at the IRS given by
\begin{equation}\label{eq:basic_optimal_angle} 
\theta ^{*}_{n,k}=\angle h_{d,k}- \angle ({h}_{1,n}+{h}_{2,k,n}),\hspace{0.5cm} n=1, {\dots }, N.
\end{equation}
\end{lemma}
\vspace{-0.2cm}
The above lemma quantifies the  gain that an IRS can offer compared to a system in the absence of IRS. However, achieving the rate in \eqref{eq:basic_optimal_rate}  requires the knowledge of the CSI through every IRS element whose complexity scales linearly with the number of elements, as stated earlier. 

Next, consider the system where the phase \textcolor{black}{of every IRS element} is selected uniformly at random from $[0, 2\pi)$ in each slot, and the user with the best PF metric is selected for transmission. We define the average rate achieved by the 
the randomly configured IRS assisted system to be $R^{(K)}= \mathbb {E}\left[\log _{2}\left(1+P|h_{k^*}|^2/\sigma^2\right)\right]$, where $k^*$ is the user selected in time slot $t$, and the expectation is taken over the randomness in the phase configuration.  
Then, under PF scheduling, \textcolor{black}{as $\tau \rightarrow \infty$, it is known that the average rate of the randomly configured IRS assisted system almost surely converges} to the average rate achievable in the beamforming configuration under 
fair resource allocation across users, i.e., \cite{Nadeem_WCL_2021} \begin{equation}\label{eq:basic_rand_2_opt}  {\lim_{K\to \infty }} \left(R^{(K)}-\frac {1}{K}\sum _{k=1}^{K} R_{k}^{BF}\right ) = 0. \end{equation}   
Note that, in \eqref{eq:basic_rand_2_opt}, the factor $\frac{1}{K}$ in the second term on left hand side accounts for the fairness ensured by the system. 

\begin{remark}[\textcolor{black}{On the convergence rate}]\label{remark1_basic_rate_of_convergence}
Let us compute the scaling of the number of users $K$ with the number of IRS elements $N$, such that, 
with a given (fixed) probability, a randomly selected phase configuration $\boldsymbol{\theta}$ at the IRS is nearly in beamforming configuration for at least one user. 
Consider an arbitrary user, and define the event $\mathcal{E}_i \triangleq \left\{ \theta_i \in \left[\theta_i^{*}-\epsilon,\theta_i^{*}+\epsilon\right] \right\}$, where $\theta_i^{*}$ is the phase angle required for the $i$th element of the IRS to be in beamforming configuration for that user. Since the phase angles at the IRS are chosen as $\theta_i \widesim[2]{\text{i.i.d.}} \mathcal{U}[0,2\pi)$, 
if we define $\mathcal{E} \triangleq \cap_{i=1}^N \mathcal{E}_i$, we have $Pr(\mathcal{E}) = \left(\epsilon/\pi\right)^N$. Then, the probability that at least one user in a $K$-user system sees an IRS phase configuration that is within $\epsilon$ distance of its beamforming configuration is 
$P_{\text{succ}} = 1 - \left(1-\left({\epsilon}/{\pi}\right)^N\right)^K.$
Hence, in order to have a fixed probability of success via i.i.d.\ randomly selected phase configurations, when $\epsilon/\pi \ll 1$, the number of users must scale with $N$ as 
\begin{equation} \label{eq:K_basic}
K \ge \left(- \log(1-P_{\text{succ}}) \right) (\pi/\epsilon)^N.
\end{equation}
In other words, the i.i.d.\ phase configuration scheme constrains the number of IRS elements that can be deployed when the number of users is limited.\footnote{\textcolor{black}{Note that the rate obtainable in an IRS assisted OC system always increases with $N$.  However, if $N$ is increased keeping $K$ fixed, the gap between the rate achieved by OC and the rate achievable under the beamforming configuration with fair resource allocation across users also increases, because the probability that no user is close to beamforming configuration increases. 
In fact, \textcolor{black}{for a large but fixed $K$ with i.i.d. channels}, the average rate in a randomly configured IRS grows as $\mathcal{O}(\log_2 N)$ (see \eqref{eq:rate_mud_reflection_div_nb}), whereas, in the beamforming configuration, it grows as $\mathcal{O}(\log_2 N^2)$ (see~\eqref{eq:basic_optimal_rate}).}} 
In the next subsections, we present and analyze schemes that improve the rate of convergence of $R^{(K)}$ to $\frac{1}{K}\sum _{k=1}^{K} R_{k}^{BF}$.
\end{remark}

\subsection{IRS-Aided Multi-user Diversity with Reflection Diversity}\label{sec:sel_div_scheme_nb}

In this subsection, we study an enhancement of foregoing scheme by offering additional \emph{reflection} diversity gain.  In this scheme, the IRS is configured using random and independent reflection coefficients during multiple consecutive pilot symbols transmitted at the beginning of each time slot. Note that, in this scheme, there is a one-to-one mapping between the pilot symbol index and the phase configuration used at the IRS.  Hence, the effective channel between the BS and the $k$th user during the $q$th pilot transmission in \textcolor{black}{a given time slot, denoted by $h_{k,q}$}, is different for each of the pilot symbols because the phase configuration of the IRS is different for each value of~$q$. 
\subsubsection{Channel Model}
We model the effective downlink channel $h_{k,q}$ using \eqref{eq:basic_channel}, with the phase configuration $\mathbf{\Theta}$ replaced with $ \mathbf{\Theta}_q$ for the $q$th pilot interval. 
\subsubsection{Scheme for IRS-Enhanced Multi-user Diversity Aided with Reflection Diversity}\label{sec:IRS_Q_scheme}
Inspired by the fast switching time of IRS phase configurations \cite{meta_material_switching_1,meta_material_switching_2} compared to the time slot (e.g., 10 ms frame duration in 5G NR~\cite{Dahlman_5g_book_2018}), we can obtain additional \textcolor{black}{reflection} diversity on top of the multi-user diversity by configuring the IRS with several random and independent reflection coefficients (phase configurations) during the pilot symbols transmitted at the beginning of each time slot. Every user chooses the best configuration among all the IRS phase configurations in every time slot, computes its PF metric, and the best user feeds back the corresponding \textcolor{black}{phase configuration index} and SNR to the BS.\footnote{In slowly varying channels, one can maintain the history of the phase configurations used in the previous time slots and the corresponding SNRs reported by the users, and avoid multiple pilot transmissions in each slot.} The BS then sets the IRS with the phase configuration received from the user selected for transmission, for the rest of the slot. 

Let $Q$ be the number of randomly chosen IRS phase configurations within a time slot, which is the same as the number of pilot transmissions. In the rest of this section, for analytical tractability, and similar to \cite{Nadeem_TWC_2021}, we consider the path loss coefficients to be equal across all links and users: $\beta_{r,k}\approx \beta_{d,k} = \beta$.\footnote{The first approximation is realistic as long as IRS path is not much longer than the non-IRS path. Also, due to the second equality, the PF scheduler boils down to the max-rate scheduler~\cite{Viswanath_TIT_2002}.} \textcolor{black}{We then have the following proposition.
\begin{proposition}\label{prop:ch_iid_Gaussian}
The effective channels, $h_{k,q}$, are i.i.d. across users and pilots for reasonably large $N$ and $Q \ll K,N$,\footnote{\textcolor{black}{More precisely, this is an approximation, obtained by assuming that the $Q$ projections result in independent random variables. This approximation would be accurate as long as $Q$ is smaller than the number of IRS elements.}} and further they follow the distribution~$\mathcal{CN}(0,\beta(N\!+1))$.
\end{proposition}
 \begin{proof}
See Appendix C in the supplementary material.
\end{proof}
We can also observe numerically (See Fig.~\ref{fig4:mud_sd_rate}, Sec.~\ref{sec.results}) that this proposition holds true even for moderate values of $N$.\footnote{\textcolor{black}{Note that, if the BS and IRS are deployed at fixed and high locations, the channel
between the BS and IRS will remain static for several time slots and have a strong LoS component. One can then model this channel using a deterministic array steering response vector~\cite{Nadeem_WCL_2021}. Since the distribution of a circularly symmetric Gaussian random variable is unchanged under rotation by a deterministic phase angle, $h_{k,q} \sim \mathcal{CN}(\cdot)$ still holds true \textcolor{black}{for any $N$}.}}}
 We now note that, as $Q$ increases, the time remaining for data transmission in each frame decreases. Thus, the average throughput of a system adopting this scheme is 
\begin{equation}\label{eq:mul_pilot_intial_rate}
    {R^{(K,Q)} =  \left(1 - \zeta Q\right) \mathbb{E} \left[\log_2\left(1+   \max_{\substack{q\in [Q], \\ k\in[K]}}   \frac{P\vert h_{k,q}\vert ^{2}}{\sigma^2}\right)\right]},  
\end{equation}
where $\left(1 - \zeta Q\right)$ is the pre-log factor accounting for the loss in the throughput due to transmitting $Q$ pilot symbols in each slot, $\zeta$ is the fraction of the time slot expended in a single pilot transmission, and the expectation is taken with respect to the random IRS phase configurations \textcolor{black}{and fading channels}. Note that we account for the $(1-\zeta Q)$ factor only in this subsection, since multiple pilot symbols are used. In rest of the paper, since only a single pilot transmission occurs, we ignore its effect on the throughput. \textcolor{black}{The following theorem characterizes the scaling of the average system throughput of IRS enhanced multi-user diversity aided with reflection diversity as a function of the system parameters. 
}
\begin{theorem}\label{thm:mulitple_pilot_rate}
Consider an $N$-element IRS aided system with $K$ users and $Q$ pilot transmissions, as described above. \textcolor{black}{Under Proposition~\ref{prop:ch_iid_Gaussian}}, the average system throughput scales~as
\begin{multline}\label{eq:rate_mud_reflection_div_nb}
     \lim_{K\to \infty } \left( R^{(K,Q)} - (1-\zeta Q) \right. \\ \left. \times \log_2\left(1+\!\frac{\beta P}{\sigma^2}(N+1)\ln(QK)\right)\right) = 0.
\end{multline}
\end{theorem}
\begin{proof}
See Appendix \ref{app:mul_pilot_txn_rate_nb}.
\end{proof}

In \eqref{eq:rate_mud_reflection_div_nb}, the pre-log factor decreases with $Q$, while the logarithmic factor increases with $Q$. 
Therefore, the exists a $Q$ for which \eqref{eq:rate_mud_reflection_div_nb} is maximized. The following lemma provides the optimal $Q$ as the solution of an implicit equation, which can be solved using fixed-point iteration methods. We skip the proof as it is straightforward.
\begin{lemma}\label{lemma_Q*}
The number of pilots $Q$ that maximizes $R^{(K,Q)}$ in \eqref{eq:rate_mud_reflection_div_nb} for a given $K$ and $N$, denoted by $\hat{Q}$, satisfies the fixed point equation
    $\log_2(QK) = 
    e^{W\left(\zeta^{-1}Q^{-1}-1\right)}/\beta(N+1),$
where $W(\cdot)$ is the Lambert $W$ function. Then, the optimal integer valued $Q$ is
\begin{equation}\label{eq:integer_optimal_q}
    Q^* = \argmax\limits_{\left \lceil{\hat{Q}}\right \rceil,\left \lfloor{\hat{Q}}\right \rfloor }  R^{(K,Q)}.
\end{equation}
\end{lemma}
\begin{remark}[\textcolor{black}{On the feedback requirement}]
The feedback requirement in this scheme is slightly higher than the previous scheme. In addition to feeding back the best overall SNR, each user also sends an additional $\log_2 Q$ bits to indicate the index of the IRS phase configuration that yielded this best SNR at the user. Furthermore, after scheduling the user by the BS, the BS has to inform the IRS to configure to the phase configuration that gave the best SNR to the scheduled user. However, this additional signalling is still substantially lower than the signalling required by conventional IRS phase optimization schemes.
\end{remark} 

\begin{remark}[\textcolor{black}{On the convergence rate}]
Continuing with  Remark~\ref{remark1_basic_rate_of_convergence}, in order to ensure that with probability at least $P_{\text{succ}}$, there is a user for which the IRS configuration used in one of the $Q$ pilots is within an $\epsilon$ ball of its optimal configuration, we need 
\begin{equation}
	K \ge  \frac{1}{Q}\left(-\log(1-P_{\text{succ}})\right) (\pi/\epsilon)^N,
\end{equation}
when $\epsilon/\pi \ll 1$. Thus, employing $Q$ random phase configurations at the IRS during the pilot transmissions is equivalent to having $KQ$ users in the system. Hence, a performance close to that achieved by  optimal configuration at the IRS is possible with fewer users compared to the scheme in Sec.~\ref{sec:IRS_basic_scheme}. 
\end{remark}
\subsection{IRS Channel Model Aware Multi-user Diversity}\label{subsec:steering_model_scheme}
In the preceding section, a method to improve the performance of the basic scheme in Sec.~\ref{sec:basic_scheme_nb} was proposed by introducing multiple pilot transmissions. However, as we will see from {Sec.~\ref{sec.results}}, for both the schemes, the gap between the optimal rate and opportunistic rate increases with the number of IRS elements, especially in the large user regime. 
We now describe a method to further overcome this limitation by accounting for the channel structure in IRS aided systems, namely, that  the IRS is deployed such  that the BS-IRS and IRS-user channels exhibit strong LoS paths. On the other hand, the direct link between the BS and user may be non-LoS and experience high path loss/shadowing effects. Thus, in this section, we ignore the contribution of the direct link, as in~\cite{Tian_EURASIP_2021}.
\subsubsection{Channel Model}\label{subsubsec:steering_model_ch}
We represent $\mathbf{h}_1$ and $\mathbf{h}_{2,k}$ as LoS channels using array steering vectors. Considering an $N$-element uniform linear array (ULA) based IRS, \textcolor{black}{the LoS channels in the sub-6 GHz bands} can be modeled as~\cite{Tian_EURASIP_2021}
\vspace{-0.2cm}
\begin{multline}\label{eq:steering_ch_model_1}
    \mathbf{h}_1 = \left[1, e^{-j\frac{2\pi d}{\lambda}\sin(\theta_A)}, e^{-j\frac{4\pi d}{\lambda}\sin(\theta_A)}, \ldots, \right. \\ \left. e^{-j\frac{2\pi (N-1) d}{\lambda}\sin(\theta_A)}\right]^T, 
\end{multline}
 \begin{multline}\label{eq:steering_ch_model_2}
     \mathbf{h}_{2,k} = h_{k}'\left[1, e^{-j\frac{2\pi d}{\lambda}\sin(\theta_{D,k})}, e^{-j\frac{4\pi d}{\lambda}\sin(\theta_{D,k})}, \ldots, \right.\\ \left. e^{-j\frac{2\pi (N-1) d}{\lambda}\sin(\theta_{D,k})}\right]^T,
\end{multline}
where $\theta_A$ and $\theta_{D,k}$ are the direction of arrival (DoA) and direction of departure (DoD) of the $k$th user at the IRS, $d$  and $\lambda$ are the inter-IRS element distance and signal wavelength, and $h_k'$ is the Rayleigh distributed channel for the $k$th user. The other parameters are as in Sec.~\ref{sec:basic_ch_model} except for the absence of the non-IRS path. For the analysis, the total path loss is considered to equal $\beta$ for all users as in~\cite{Nadeem_TWC_2021}. 
\subsubsection{Scheme for IRS  Channel Model Aware Multi-user Diversity}
Since the locations of the IRS and BS are fixed, and owing to the strong LoS path between the BS and IRS, it is  realistic to consider that the channel $\mathbf{h}_1$ remains constant for relatively long time intervals~\cite{Tse_book_2005}. Hence, we assume that the knowledge of $\mathbf{h}_1$ is available at the IRS. For example, we can perform a system calibration to estimate the channel $\mathbf{h}_1$ offline.\footnote{One approach to accomplish this is illustrated in \cite{Wei_1_ICL_2021}. The authors rely on an active receiver near the IRS, whose channel to the BS is similar to that of the BS-IRS channel. Then, estimating the channel between the BS and this active receiver can help in estimating~$\mathbf{h}_1$.} Further, in slow fading scenarios, the DoD statistics of all the users (at the IRS) do not change significantly over multiple time slots. Thus, by means of this historic information, the BS could potentially initialize the IRS with the range of DoDs to all the users. The IRS then draws the phase angles from a distribution that depends on the DoD statistics. We first derive a distribution from which IRS phase angles can be randomly drawn. With $\theta'_k \triangleq \frac{2\pi d}{\lambda}(\sin(\theta_A) + \sin(\theta_{D,k}))$, the channel at user $k$ for \textcolor{black}{the} IRS configuration $\boldsymbol{\Theta}$ is given by
\begin{align}\label{eq:channel_steering_inner_pdt}
& h_k = \sqrt{\beta} \mathbf{h}_{2,k}^T \boldsymbol{\Theta}\mathbf{h}_1  \\
& = \sqrt{\beta} h_k' \sum_{n=1}^N e^{-j\left( (n-1)\theta'_k  \right) + j\theta_n}. \nonumber
\end{align} 
Clearly, due to the Cauchy-Schwarz inequality, $|h_k|$ is maximized iff $\theta_i = \frac{2\pi(i-1) d}{\lambda}  (\sin(\theta_A) + \sin(\theta_{D,k}))$ for all $i$, and also denotes the beamforming configuration of the IRS. 

\indent Let the DoA at the IRS from the BS be $\theta_A$ and let the DoDs at the IRS to the users be \textcolor{black}{randomly and independently generated from a uniform distribution with support $[\phi_0,\phi_1]$}. Then, in every time slot, the $i$th element of the IRS randomly chooses its phase configuration $\theta_i$ as follows:
\begin{equation}\label{eq:steering_optimal_ditbn}
	\theta_i = \frac{2\pi  (i-1) d}{\lambda}(\sin(\theta_A) + \sin(\phi)),
\end{equation}
where $\phi \sim \mathcal{U}[\phi_0,\phi_1]$. The rest of the scheme proceeds as in Sec.~\ref{sec:IRS_basic_scheme}, with the BS scheduling the user with the highest PF metric for data transmission in the current~slot. 

To investigate the performance of this scheme, first, using \eqref{eq:channel_steering_inner_pdt}, it is clear that
\begin{equation}\label{eq:steer-ch-gain}
	|h_k|^2 = \beta {\left| \sum\nolimits_{n=1}^{N} e^{-j\left((n-1)\theta'_k-\theta_n \right) } \right|}^2 \cdot |h_k'|^2.
\end{equation} 
The maximum value of the first term in \eqref{eq:steer-ch-gain} is $\beta N^2$ which is achieved by the beamforming configuration. Thus, when $K$ is large, for every $\eta \in (0,1)$, there exists a $\delta>0$ such that, for a subset of $\eta K$ users,  almost surely, we have~\cite[Sec.III.B]{Viswanath_TIT_2002} 
\begin{equation}\label{eq:min_rate_ch_aware}
\beta {\left| \sum\nolimits_{n=1}^{N} e^{-j\left((n-1)\theta'_k-\theta_n \right) } \right|}^2 > \beta N^2 - \delta.
\end{equation}
\textcolor{black}{Thus, when $K$ is large, for any randomly chosen IRS phase configuration as per~\eqref{eq:steering_optimal_ditbn}, there will almost surely exist a set of users whose overall channel experiences near-optimal beamforming configuration.} On the other hand, the second term in \eqref{eq:steer-ch-gain} denotes the square of the channel gains, which are i.i.d.\ across users. We characterize the behavior of this term using extreme value theory (EVT). In particular, using Lemma~\ref{EVT_1}, it can be shown that $\max_k |h_k'|^2$ grows as $\ln K$. Hence, among the $\eta K$ users, the maximum of $|h_k|^2$ grows with $K$ at least as fast as
\begin{equation}\label{eq:steer-growth}
(\beta N^2 - \delta)\ln(\eta K) = (\beta N^2 - \delta)\ln(K) + \mathcal{O}(1),
\end{equation}
as $K \rightarrow \infty$. Clearly, the case with $\delta = 0$, which happens when at least one user is in beamforming configuration, serves as an upper bound on the rate of growth of the $|h_k|^2$ in \eqref{eq:steer-growth}. As a consequence, we have the following theorem.

\begin{theorem} \label{prop:rate_irs_channel_aware}
For the IRS channel model aware multi-user diversity scheme, the average system throughput scales as
 \begin{equation}\label{eq:rate_steering_scheme}
\lim_{K\to \infty}  \left( R^{(K)} - \mathcal{O}\left( \log_2\left(1 + \frac{\beta P}{\sigma^2}N^2 \ln K\right)\right) \right) = 0.
\end{equation}
\end{theorem}
The term $N^2$ in \eqref{eq:rate_steering_scheme} shows that this scheme attains the  maximum possible array gain from the IRS  by exploiting the presence of strong LoS paths, while 
the $\log K$ term is due to multiuser diversity. This scheme thus even outperforms the scheme where  optimization methods are used in IRS aided systems without multiuser diversity (e.g., see \eqref{eq:basic_optimal_rate}). \textcolor{black}{Also, the SNR scaling under i.i.d.\ channels is $\mathcal{O}(N)$, whereas it is $\mathcal{O}(N^2)$ under strong LoS channels. This is because, under i.i.d.\ channels, the variance of the effective channel  scales as $N$ (see Proposition~\ref{prop:ch_iid_Gaussian}), while in the latter case it scales as $N^2$, at least for the subset of users satisfying \eqref{eq:min_rate_ch_aware}. In turn, when the scheduler the best user for data transmission, the SNR scales as $\mathcal{O}(N^2)$ as per~\eqref{eq:rate_steering_scheme}. A similar observation is made in \cite{Viswanath_TIT_2002} in the non-IRS context, when comparing the performance of i.i.d. fast fading channels and correlated channels.}
\begin{remark}[\textcolor{black}{On the convergence rate}]
Similar to Remark \ref{remark1_basic_rate_of_convergence}, under the channel model in \eqref{eq:steering_ch_model_1},  \eqref{eq:steering_ch_model_2}, we have $Pr(\mathcal{E}) = 1-\left(1-\left(\frac{\epsilon}{\pi}\right)\right)^K$ \textcolor{black}{when the IRS phases are sampled as in~\eqref{eq:steering_optimal_ditbn}}. Thus, the $K$ required for near-optimal beamforming does \emph{not} grow with $N$, and the opportunistic rate converges much faster to the beamforming based rate compared to the scheme in Sec. \ref{sec:basic_scheme_nb}. This is illustrated in Figs.~\ref{fig5:steering_rate} and~\ref{fig6:rate_vs_IRS} later in the sequel.
\end{remark}
\vspace{-0.7cm}
\textcolor{black}{
\begin{remark}[mmWave bands]
	The channel model in~\eqref{eq:steering_ch_model_1}, \eqref{eq:steering_ch_model_2} is a  special case of mmWave channels~\cite{Raghavan_JSTSP_2016} with the number of paths set to $1$. 
	Since, by exploiting the knowledge of the channel statistics to design the distribution from which the random phase configurations are drawn we can obtain significant multi-user diversity gains even with a relatively small number of users, randomly configured IRS-aided OC can obtain significant benefits in mmWave scenarios also. 
\end{remark}}
\color{black}
\section{Single IRS Assisted Opportunistic User Scheduling  for Wideband Channels}\label{sec.opp_schemes_wideband}
In this section, we investigate IRS assisted OC over an $L$ tap wideband channel.
 We consider a multiuser OFDM system where all users in the system are served over a given total bandwidth. Since the IRS operates over the entire bandwidth  (i.e., it is not possible to apply different phase configurations for different sub-bands),\footnote{This assumes that the IRS elements are not frequency selective, similar to past work in the area \cite{Zheng_WCL_2020,Yang_TCOM_2020}. In fact, by appropriately designing the tuning parameters of the IRS circuit elements, it is possible to achieve non-frequency selectivity of the IRS elements even in wideband systems \cite{Katsanos_Arxiv_2022}.} we first analyze the performance of an IRS assisted OFDM system where all the subcarriers are allocated to a single user who has the best channel condition collectively among the subcarriers. In the second scheme, we configure the OFDM based multiple access (OFDMA) and study the performance improvement offered by the multiplexing gain in addition to multi-user diversity. We refer the former scheme as single-user OFDM (SU-OFDM) and latter scheme as OFDMA. As before, for the  analysis, we assume that all users experience similar large scale propagation effects with path loss coefficient $\beta$, and hence we model the channels across the users in an i.i.d. fashion. 
\subsection{IRS Enhanced Multi-user Diversity in a Single-user OFDM (SU-OFDM) System}\label{sec.su-ofdm}
\subsubsection{Channel Model}\label{sec:ch_model_ofdm}
Consider a time domain channel  seen by the $k$th user in an $N$-element IRS setting. 
Let $\mathbf{h}_{d,k} \in \mathbb{C}^{L\times 1}$ be the $L$-tap channel between the BS and user $k$ through the direct (non-IRS) path. Let $\mathbf{H}_{2,k} \in \mathbb{C}^{N \times L}$ denote the $L$-tap channel between IRS and user $k$ across all IRS elements. Note that, without loss of generality, we assume that the number of taps in the direct channel and the IRS-user channel to be the same. This can be done by letting $L$ denote the maximum of the number of taps in the two channels. Since the channel between the BS and IRS is typically LoS, it can be modeled as a single-tap channel between the BS and each of the $N$ elements of the IRS, denoted by $\mathbf{h}_1 \in \mathbb{C}^{N \times 1}$ (see~\cite{Rui_TWC_2020}.)
Furthermore, due to the strong LoS component, $\mathbf{h}_1$ can be modelled as an array steering response vector when the IRS is configured as a ULA (see \eqref{eq:steering_ch_model_1}). 
We assume that channels between the IRS and the users across all the $L$ taps are independent of each other~\cite{Nadeem_WCL_2021,Tse_book_2005}.
The exact statistics of the channels are provided below. 
The composite channel of user $k$ can then be compactly written as
\begin{equation}\label{eq:ofdm_ch}
\mathbf{h}_k = \mathbf{h}_{d,k} + \mathbf{H}_{2,k}^T\boldsymbol{\Theta}\mathbf{h}_1 \in \mathbb{C}^{L \times 1}.
\end{equation}
 
In this work, we use an exponentially decaying power delay profile (PDP) in the lag domain. 
Let $\breve{h}_{k,l,n} \triangleq h_{1,n}h_{2,k,l,n}$ denote the gain of the $l$th tap of the fading channel between the BS and the $k$th user through the $n$th IRS element. Then, the PDP of the link is given by
\begin{equation}
    a_{l} \triangleq \mathbb{E}[|\breve{h}_{k,l,n}|^2] = ce^{- \nu l/L}, \hspace{0.1cm} \forall k \in [K], n \in [N], \label{eq:pdp}
\end{equation} where $c$ is chosen such that $\sum_l \mathbb{E}[|\breve{h}_{k,l,n}|^2] = 1$, and $\nu$ captures the decay rate of the channel tap power with $l$. Hence, we have, $\|\mathbf{a}\|_1 = 1$, where $\mathbf{a} \triangleq [a_{1}, a_{2}, \ldots, a_{L}]^T$ represents the power in each of the $L$ taps. Therefore, the $l$th component of the channel in \eqref{eq:ofdm_ch} can be written as 
$h_{k,l} = h_{d,k,l} + \sum_{i=1}^{N} e^{j\theta_i}\breve{h}_{k,l,i}$. If $h_{d,k,l},h_{2,k,l,i} \sim \mathcal{CN}(0,a_l)$ across the IRS elements and since $|h_{1,n}|^2 = 1$ for all $n \in [N]$, it is easy to show that $h_{k,l} \sim \mathcal{CN}(0,(N+1)a_l)$ and independent across the users and $L$ taps. Equivalently, in the OFDM system with $M$ subcarriers, if we let $\tilde{\mathbf{h}}_{k} \in \mathbb{C}^{M \times 1}$ denote the frequency-domain channel vector for user $k$, we have $\tilde{\mathbf{h}}_{k} = \mathbf{F}_{M,L}\mathbf{h}_k$ where $\mathbf{F}_{M,L}$ is the matrix containing the first $L$ columns of the $M \times M$ DFT matrix.\footnote{In this paper, we compute the discrete Fourier transform (DFT) as $X[m] = \sum_{l=0}^{M-1}x[l]e^{-j\frac{2\pi ml}{M}}$ for all $m \in \{0,1,\ldots, M-1\}$.} 
Thus, the channel at subcarrier $m$ for user $k$ follows $\tilde{h}_{k}[m] \sim \mathcal{CN}\left(0,N+1\right)$, and we also have the Parseval's relation $\mathbb{E}[\|\tilde{\mathbf{h}}_{k}\|_2^2] = M \, \mathbb{E}[\|\mathbf{h}_k\|_2^2]$.

\subsubsection{Scheme for IRS-Enhanced Multi-user Diversity in SU-OFDM Systems}
As before, we randomly set the phase configuration at the IRS in every time slot. The BS then applies equal power on all the subcarriers and broadcasts pilot symbols to all the users. In this section, for simplicity and analytical tractability, we assume that the BS uses equal power allocation across subcarriers during data transmission also; note that this is near-optimal in the high SNR regime.  
The users estimate the channels and compute the sum rate obtainable by them across all the subcarriers, and compute their respective PF metrics. The user with the highest PF metric sends its identity back to the BS, and  is scheduled for transmission using all the subcarriers by the BS. Note that, with a slightly higher feedback overhead, the scheme easily extends to the case where optimal water-filling power allocation is used by the BS. Here, after estimating the channels across the subcarriers, the user computes the sum rate achievable by it with water-filling power allocation, and uses this to compute its PF metric. In this case, instead of feeding back a packet containing the user identity and the SNR, the user with the highest PF metric will also need to include the power allocation vector, which is of size $\approx 4M$ bits (assuming $16$-levels of power control in each subcarrier.)

Under equal power allocation, in a $K$ user system, the maximum average sum rate obtainable across subcarriers with a total power constraint $P$ and noise variance $\sigma^2$ is
\begin{align}\label{eq:su_ofdm_intial_rate}
\hspace{-0.2cm} R_{\text{SU-OFDM}}^{(K)} = \max\limits_{1 \leq k \leq K} \sum_{m=0}^{M-1} \log_2\left(1 + \frac{\beta P}{M\sigma ^2}|\tilde{h}_k[m]|^2 \right)\!.\!
\end{align}
\textcolor{black}{From~\eqref{eq:su_ofdm_intial_rate}, it is clear that we need to characterize the maxima of the sum of random variables. However, since the expression is not easily tractable, we upper bound the term by invoking the Jensen's inequality. From Parseval's theorem, $\sum_{m=0}^{M-1} |\tilde{h}_k[m]|^2 = M \sum_{l=0}^{L-1}|h_{k,l}|^2$. Using the monotonicity of $\log$, we get 
\begin{equation}\label{eq:su-ofdm-jensen-parsevel}
	R^{(K)}_{\text{SU-OFDM}} \leq \log_2\left(1+ \frac{\beta P}{\sigma^2} \left\{\max\limits_{1 \leq k \leq K} \sum_{l=0}^{L-1}|h_{k,l}|^2\right\}  \right).
\end{equation}
We illustrate the tightness of the above upper bound in the numerical results section.
Recall that 
$|h_{k,l}|^2 \sim \exp(\frac{1}{(N+1)a_l})$ and these form a set of independent and non-identically distributed (i.n.d.) random variables across $L$ taps. First, we characterize the distribution of the sum-term in~\eqref{eq:su-ofdm-jensen-parsevel}. We can show from~\cite{Sheetal_TCOM_2012} that if $\{X_i\}_{i=1}^{L}$ are a set of $L$ i.n.d. exponential random variables with mean $\mu_i$, then the cumulative distribution function (cdf) of $Y \triangleq \sum_{i=1}^{L} X_i$ is given by 
\begin{multline}\label{eq:finite_sum_exponential_cdf}
	F_{Y}(y)={y^{L}\over \Gamma(1+L) \prod_{i=1}^{i=L} \mu_i^{-2}}\times \\  \Upsilon_{2}^{(L)} \Bigg(1, \ldots, 1;1+L;-{\mu_{1}y}, \ldots, -{\mu_{L}y}\Bigg),	
\end{multline}
where $\Upsilon_2^{(L)}(\cdot)$ is the confluent Lauricella function~\cite{Exton_book_1976}. Thus, setting $\mu_i = (N+1)a_i$ in~\eqref{eq:finite_sum_exponential_cdf} will give the distribution of the sum-term in~\eqref{eq:su-ofdm-jensen-parsevel} and call it $\Tilde{F}(\cdot)$. In what follows, we characterize the maximum of such i.i.d. sum-terms. To that end, we can show that the cdf in~\eqref{eq:finite_sum_exponential_cdf} satisfies the Von Mises' condition (see Lemma~\ref{EVT_2} in Appendix~\ref{app:su_ofdm_rate})~\cite{Sheetal_TCOM_2012}. Thus, 
\begin{equation}\label{eq:su-ofdm-exact-characterization}
	\max\limits_{1 \leq k \leq K} \sum_{l=0}^{L-1}|h_{k,l}|^2 \xRightarrow[]{K \rightarrow \infty} \Tilde{F}\left(1-\dfrac{1}{K}\right),
\end{equation}
where, by $\{X_k\}_{k=1}^{K} \xRightarrow[]{K \rightarrow \infty} c$, we mean $\lim_{K\rightarrow\infty}X_k - c \xrightarrow[]{d} Y$ and $Y$ is a degenerate random variable. In other words, the sum-term in~\eqref{eq:su-ofdm-jensen-parsevel} can be replaced with the right hand side of~\eqref{eq:su-ofdm-exact-characterization} when $K$ is large. However, the resultant expression, although accurate, has two demerits: $1)$ It does not provide any explicit and useful insight on how the sum-rate scales with the total number of users, $2)$ the characterization is not  tractable for comparison and analyzing the performance.}

\textcolor{black}{Hence, we seek approximations to the system model by considering large $L$ (In Sec.VII.B of the supplementary material, we discuss on how large $L$ needs to be.) In the next section (Sec.\ref{sec.results}), we numerically show that this approximation works well even when $L$ is as small as $5$.}
Then, in view of~\eqref{eq:su-ofdm-jensen-parsevel}, we have the following theorem for reasonably large~$L$.
\begin{theorem}\label{thm:su_ofdm_rate_scale}
Consider an $N$-element IRS assisted SU-OFDM system with $M$ subcarriers and $L$ time-domain taps with  power delay profile $\mathbf{a}$ and a total power constraint $P$ and noise variance $\sigma^2$. Then, for large $L$, the average sum rate of IRS enhanced multi-user diversity in an SU-OFDM system \textcolor{black}{under equal power allocation}, $R_{\text{SU-OFDM}}^{(K)}$, scales as
\begin{multline}\label{eq:su-ofdm-rate-law}
{\lim_{K\to \infty } } \left(R_{\text{SU-OFDM}}^{(K)} - \mathcal{O}\left(\log_2\left\{1 + \frac{\beta P}{\sigma^2}(N+1) \right.\right.\right.\\\left.\left.\left. \times \left[1 + \| \mathbf{a}\|_2 \Phi^{-1}\left(1-\frac{1}{K}\right)\right]   \right\}\right)\right) = 0,
\end{multline}
where $\Phi^{-1}(\cdot)$ is the inverse of the cdf of a standard normal random variable. 
\end{theorem}
\begin{proof}
See Appendix \ref{app:su_ofdm_rate}.
\end{proof}
In the above result, since the argument of $\Phi^{-1}(\cdot)$ is close to 1 for large $K$, we can use $\Phi(x) \approx \frac{1}{2}\left( 1 + \sqrt{1-e^{-\frac{2}{\pi}x^2}}\right)$~\cite{Eidous_MS_2016}. Consequently, from~\eqref{eq:su-ofdm-rate-law}, we can explicitly determine the dependence of the sum rate on $K$, as in the following corollary. 
\begin{corollary}
For the setup in Theorem ~\ref{thm:su_ofdm_rate_scale}, we have
 \begin{multline}\label{eq:su_ofdm_rate_approx}
{\lim_{K\to \infty }} \left(R_{\text{SU-OFDM}}^{(K)} - \mathcal{O}\left(\log_2\left\{1 + \frac{\beta P}{\sigma^2}(N+1) \right.\right.\right.\\\left.\left.\left. \times \left[1 +  \Vert\mathbf{a}\Vert_2 \sqrt{\frac{\pi}{2}\ln K}\right]   \right\}\right)\right) \approx 0.
\end{multline}
\end{corollary}
Comparing the SU-OFDM performance given by the above equation with the performance in narrowband channels (see \eqref{eq:rate_mud_reflection_div_nb}, with $Q=1$), the main difference in the multi-carrier case is the presence of $\|\mathbf{a}\|_2$ and the dependence on the number of users as $\sqrt{\ln K}$ instead of $\ln K$. The $\sqrt{\ln K}$ dependence is a consequence of the upper-bounding technique used to obtain \eqref{eq:su_ofdm_rate_approx}; and apart from the $\|\mathbf{a}\|_2$ factor, the performances are similar in the two cases. 


\subsection{IRS Enhanced Multi-user Diversity in OFDMA Systems}

We now consider 
an OFDMA system, where, instead of allotting all the subcarriers to one of the users, each subcarrier is allotted to a single, possibly different user. On the other hand, a given user can be allotted one or more subcarriers.
We consider the same channel model as in Sec. \ref{sec:ch_model_ofdm}.

\subsubsection{Scheme for IRS-Enhanced Multi-user Diversity in OFDMA Systems}
The scheme is similar to the single-user OFDM system in Sec.~\ref{sec.su-ofdm}, except that the user scheduling is done on a per subcarrier basis instead of allotting  all the subcarriers to the user with the best sum rate across subcarriers. Recall that the channel coefficient of the $m$th subcarrier of user $k$ is denoted by $\tilde{h}_k[m]$. Then, the average sum rate under equal power allocation in the IRS enhanced OFDMA based multi-user diversity scheme is given by
\begin{equation}\label{eq:su-ofdm-rate-scale-intial}
R^{(K)}_{\text{OFDMA}} = \sum_{m=0}^{M-1} \log_2\left\{1+ \frac{\beta P}{M\sigma^2} \max\limits_{1 \leq k \leq K} |\tilde{h}_k[m]|^2 \right\}.
\end{equation} Thus, we have the following theorem that characterizes the average sum rate of an OFDMA system. 
\begin{theorem}\label{thm:ofdma_rate_law}
Consider an $N$-element IRS assisted OFDMA system with $M$ subcarriers,  a total power constraint $P$, and noise variance $\sigma^2$. Then, the average sum rate exploiting multi-user diversity \textcolor{black}{under equal power allocation}, $R_{\text{OFDMA}}^{(K)}$, scales~as
\begin{equation}\label{eq:ofdma-rate-law}
{\lim_{K\to \infty }} \!\!\left\{\!R^{(K)}_{\text{OFDMA}} - \!M\log_2\left(1+ \frac{\beta P}{M\sigma^2}(N+1)\ln K\right) \right\} = 0.
\end{equation}
\end{theorem}
\begin{proof}
The key step is to characterize the random variable $\max\nolimits_{1 \leq k \leq K} |\tilde{h}_k[m]|^2$ for large $K$. Since $|\tilde{h}_k[m]|^2 \widesim[2]{\text{i.i.d.}} \exp(1/(N+1))$, we can apply Lemma \ref{EVT_1}  to obtain the scaling law in \eqref{eq:ofdma-rate-law} in the same way as derived in Appendix~\ref{app:mul_pilot_txn_rate_nb}.
\end{proof}
\begin{remark}[\textcolor{black}{On the performance of OFDMA and SU-OFDM}] \label{rem:OFDMvsOFDMA}
The IRS assisted OFDMA scheme  outperforms the SU-OFDM scheme due to two reasons: 1) Since there are $M$ parallel  channels in the OFDMA scheme, and this offers additional selection/frequency diversity gain over  an SU-OFDM system. 
  2) While a given user may see different channel coefficients on  different subcarriers, the IRS configuration is common across all subcarriers. Thus, even in the asymptotic number of users, it is not possible for any user to be in  beamforming configuration on all the subcarriers in an SU-OFDM scheme. On the other hand, in the OFDMA scheme, since the setup boils down to the availability of $M$ parallel  channels, and when $K$ is large, it is possible for the IRS to be close to the beamforming configuration with high probability on all subcarriers, by scheduling different users on the different subcarriers. 
However, the feedback overhead in the OFDMA scheme is $M$ times that of SU-OFDM, since the BS needs to find the best user on each of the $M$ subcarriers. Note that, with OFDMA, one can still use a low feedback overhead timer- or splitting-based scheme \cite{suresh_NCC_2010,Shah_TCOM_2010} for identifying the best user to schedule for data transmission, but on a subcarrier-by-subcarrier basis. 
\end{remark}

\section{Numerical Results}\label{sec.results}
In this section, we validate the analytical results derived as well as quantify the relative performance of the different schemes proposed in the previous sections, through Monte Carlo simulations.  A single antenna BS is located at $(0,0)$ (in metres), the IRS is at $(0,250)$ and single antenna users are uniformly distributed in the rectangular region with diagonally opposite corners $(100,500)$ and  $(500,1000)$. The path losses are computed as $\beta = 1/d^\alpha$ where $d$ is the distance and $\alpha$ is the path loss exponent. We use $\alpha=2,2.8$ and $3.6$ in the BS-IRS, IRS-user and BS-user (direct) links, respectively \cite{Nadeem_WCL_2021}. 
\textcolor{black}{Further, we consider a BS transmitting with power $P = -10$~dBm and noise variance at the receiver $\sigma^2 = -117.83$~dBm, corresponding to a signal bandwidth of $400$~kHz at a temperature of $300$~K}.
Then, in the absence of the IRS, a user at the point closest to the BS experiences an average SNR of about $10.3$~dB, while the farthest user experiences an average SNR of $-1.9$ dB. 
The fading channels are randomly generated as per the distributions discussed in the previous sections.

We first evaluate the performance of the scheme described in Sec. \ref{sec:basic_scheme_nb}. In Fig. \ref{fig2:basic_scheme}, we plot the average throughput offered by a randomly configured IRS-assisted OC scheme operated using a proportional fair scheduler with $\tau = 5000$.  We compare the performance of the OC scheme against that of the beamforming-optimal scheme, given by \eqref{eq:basic_rand_2_opt}. The throughput of the OC system improves with the number of IRS elements and users in the system. On the other hand, the gap between the throughput of the OC scheme and that of the optimally configured IRS based scheme also increases with the number of IRS elements, in line with our discussion in Remark \ref{remark1_basic_rate_of_convergence}. We also see that the IRS assisted system significantly outperforms opportunistic scheduling in the absence of the IRS, when the BS is equipped with $N$ antennas~\cite[Sec.~III.A]{Viswanath_TIT_2002} (i.e., the same as the number of IRS elements used.)\footnote{The BS uses weights $\alpha_n(t)$ and phase $\theta_n(t)$ at the $n$th antenna at time $t$, leading to $h_{k}(t) = \sum_{n=1}^{N} \,\sqrt{\alpha_n(t)} \,e^{j\theta_n(t)}h_{nk}(t)$ as the effective channel gain. At each $t$,  $\alpha_n(t)$ and $\theta_n(t)$ are set randomly.} 
This is because the BS is constrained by a total radiated power. Hence, increasing the number of antennas reduces the transmit power per antenna, and results in a throughput that improves only marginally with the number of antennas at the BS. 
On the other hand, since the IRS uses passive reflective elements, the total received power at the user increases quadratically with the number of IRS elements under the optimal beamforming configuration (see, e.g., Theorem~\ref{prop:rate_irs_channel_aware} or \cite{Wu_GLOBECOM_2018}).

\begin{figure}[t]
\includegraphics[width=\linewidth]{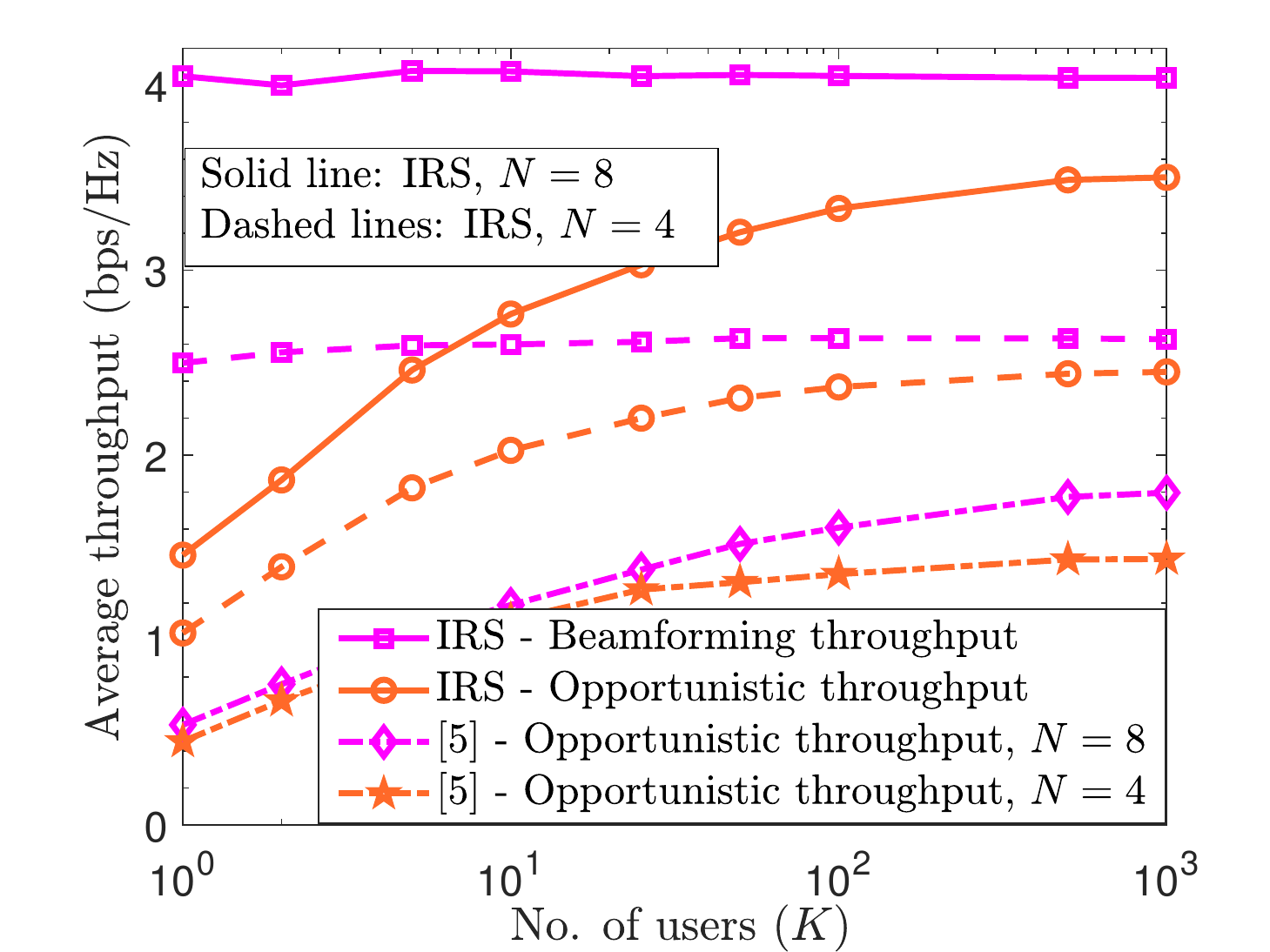}
\centering
\caption{Average throughput vs. the number of users compared with the opportunistic scheme in~\cite[Sec.III.A]{Viswanath_TIT_2002}.}
\label{fig2:basic_scheme}
\end{figure}

Next, in Figs. \ref{fig3:q_vary} and \ref{fig4:mud_sd_rate}, we evaluate the scheme in Sec.~\ref{sec:sel_div_scheme_nb}. In this experiment, we use $\zeta = 0.01$ in \eqref{eq:rate_mud_reflection_div_nb}.\footnote{\textcolor{black}{We note that $\zeta$ depends on the coherence time of the channel. It has been shown in~\cite{Lozano_EWC_2010} that an optimal $\zeta$ is around $0.01$ for channels with moderate fading rate operating at moderate SNR, in single antenna systems.}} Hence, the rate goes to zero when $Q=100$, since no symbols are left for data transmission. In Fig. \ref{fig3:q_vary}, we plot the throughput as a function of $Q$, the number of pilot transmissions. The optimal $Q^*$ that yields the best trade-off between the pilot overhead and \textcolor{black}{reflection} diversity gain, given by \eqref{eq:integer_optimal_q}, agrees with the integer $Q$ at which the throughput achieves its maximum. \textcolor{black}{In Fig.~\ref{fig4:mud_sd_rate}, we  compare the performance in terms of the achievable system throughput obtained using $Q=1$ (``\text{Opportunistic throughput, $Q=1$}") against that obtained by using $Q=Q^*$ pilot transmissions (``\text{Opportunistic throughput, $Q=Q^*$}''.) The figure  shows the additional gain due to the reflection diversity, particularly when the number of users is small. The gain is marginal when $K$ is large, partly because the throughput depends weakly on $Q$ since it scales as $\log\big(\ln (QK)\big)$, and partly because $Q^*$ itself reduces with $K$. In the same figure, we also plot the performance of the IRS assisted opportunistic system under equal path loss across users, and see that the simulations (``\text{Equal path loss - opp. throughput, $Q=Q^*$}") match with the theoretical result in Theorem~\ref{thm:mulitple_pilot_rate} (``\text{Theorem $1$, $Q=Q^*$}''.)}

\begin{figure}[t]
\includegraphics[width=\linewidth]{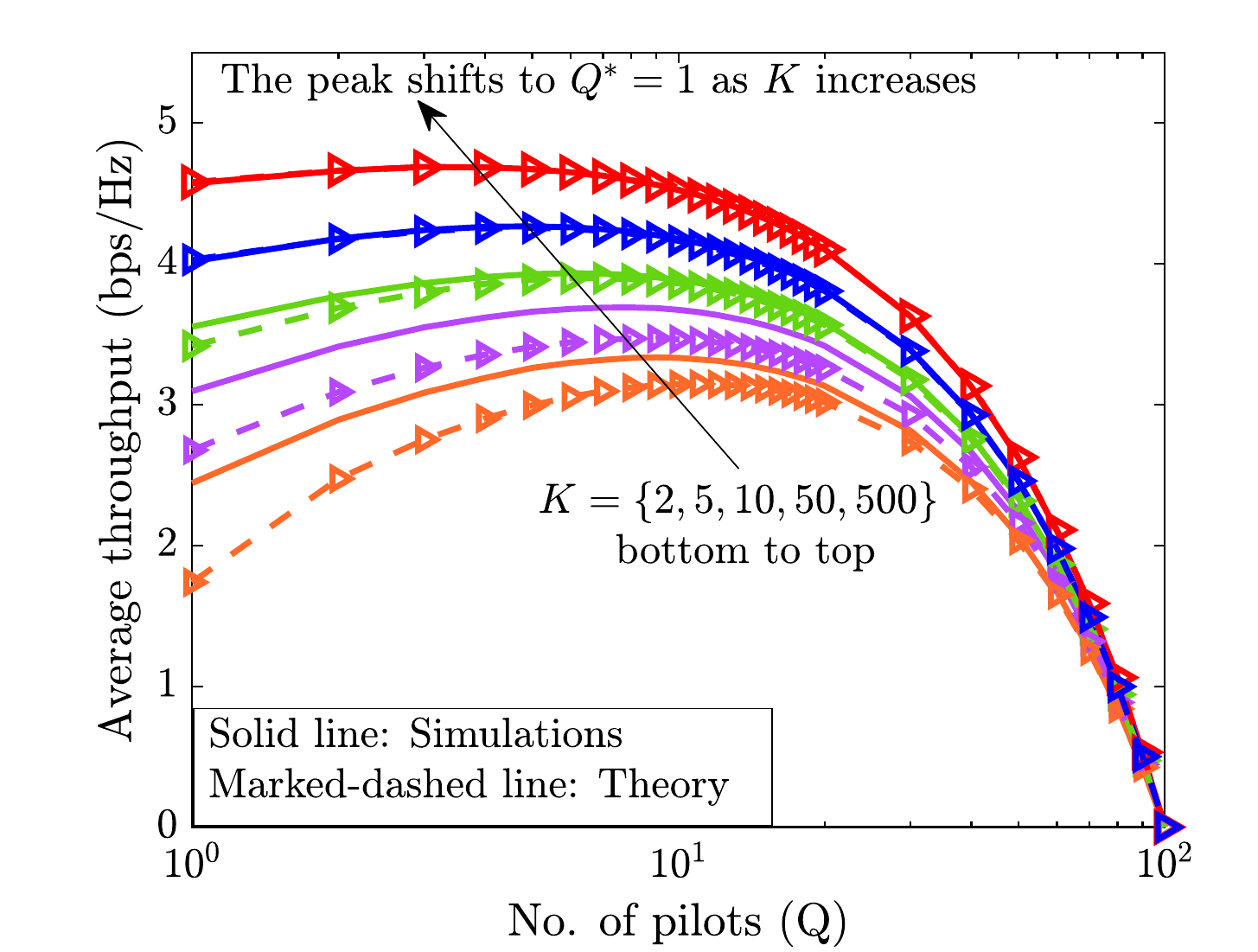}
\caption{Average throughput as a function of the number of pilot transmissions, for $N=8$.}
\label{fig3:q_vary}
\end{figure}
\begin{figure}[t]
\includegraphics[width=\linewidth]{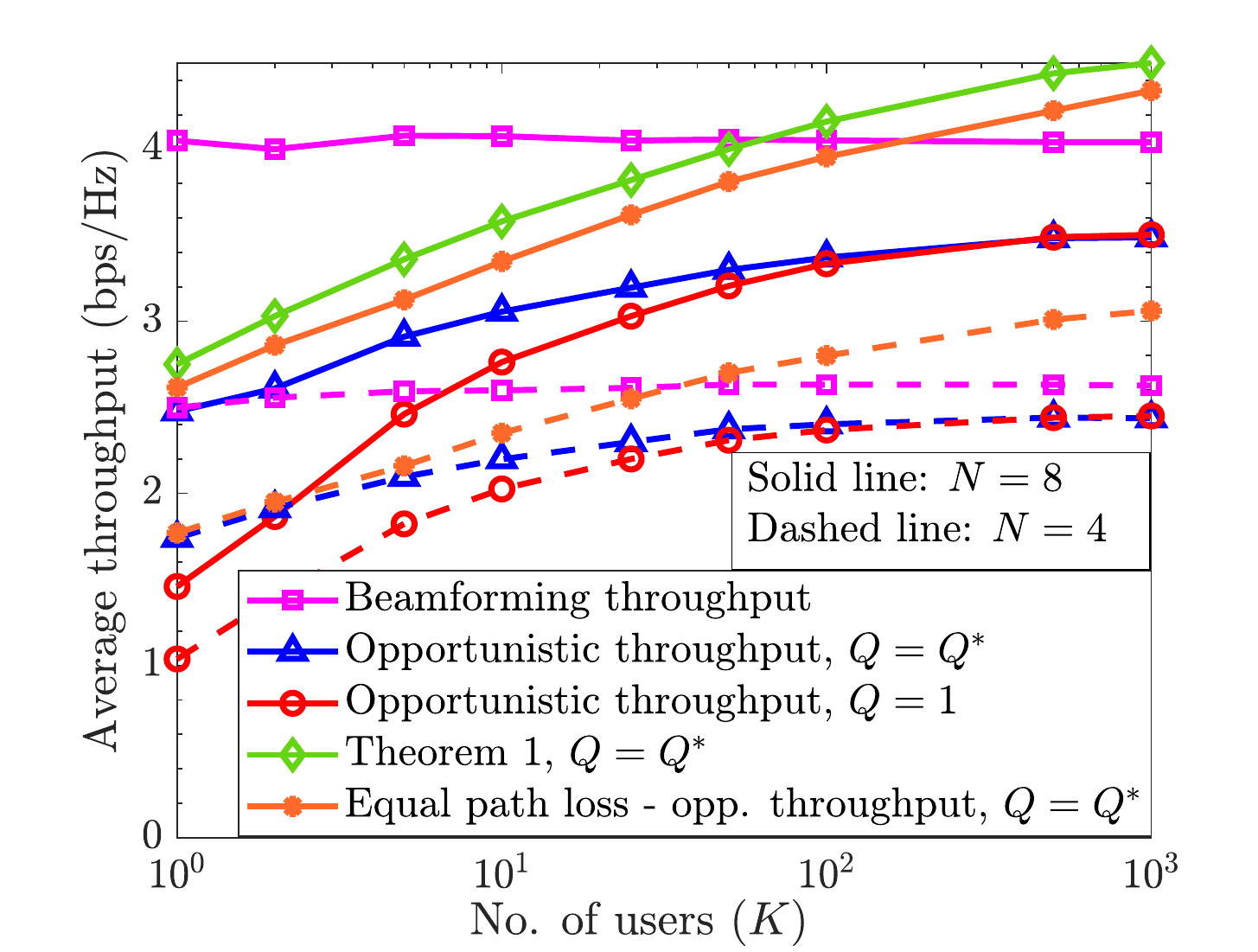}
\caption{Average throughput as a function of the number of users, with $Q^*$ pilot transmissions.}
\label{fig4:mud_sd_rate}
\vspace{-0.1cm}
\end{figure}

In Fig. \ref{fig5:steering_rate}, we  study the performance of \textcolor{black}{the channel model aware OC scheme}
as in Sec.~\ref{subsec:steering_model_scheme} and the IRS draws phase angles randomly as per \eqref{eq:steering_optimal_ditbn}, \textcolor{black}{with users' DoDs (at the IRS) being randomly and independently sampled from a uniform distribution in $[-40^\circ,40^\circ]$ with $\theta_A = 20^\circ$. Further, the users are located in the region as mentioned in paragraph~$1$ of this section}. We also use the uniform distribution to draw the phase angles at all the IRS elements independently, and show its performance in the same figure. We see that the performance of the OC system  improves dramatically when channel model and DoD statistics at the IRS are used in selecting the phase configurations. 

\begin{figure}[!t]
\includegraphics[width=\linewidth]{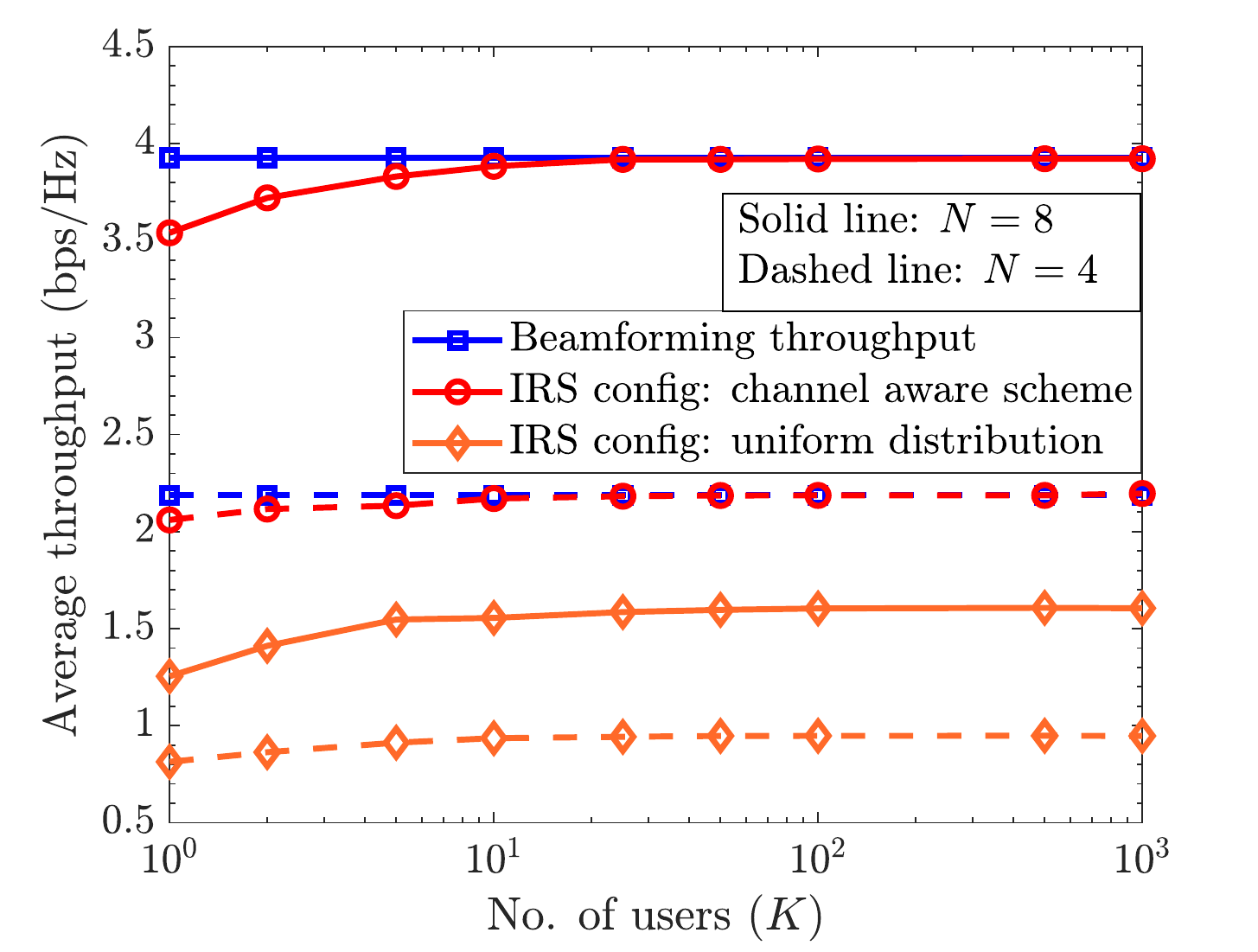}
\caption{Average throughput as a function of the number of users, for channel model aware scheme.}\label{fig5:steering_rate}
\vspace*{-0.1cm}
\end{figure}

Next, in Fig. \ref{fig6:rate_vs_IRS}, we investigate the performance gap between the randomly configured IRS based OC and the rate obtained from the beamforming configuration, as a function of number of IRS elements with and without the knowledge of DoD statistics at the IRS. The different curves correspond to the system having varying number of users. We see that, for a given number of users, the performance of the system with the IRS phase angles drawn exploiting the knowledge of the DoD statistics is very close to the performance of an optimized IRS even if the number of IRS elements is as large as $1024$. Furthermore, in this regime, even when the number of users in the system is as small as $50$, the OC performance is still close to the coherent beamforming rate. On the other hand, the performance of an IRS assisted OC scheme with the phase angles drawn independently from the uniform distribution becomes increasing worse relative to the beamforming rate as the number of IRS elements increases. Moreover, the effect of multi-user diversity is hardly evident in the latter case, when the number of IRS elements is large. In a nutshell, for a given number of users, the \textcolor{black}{channel model aware scheme} offers two benefits: 1) The rate  of the OC system remains close to the optimal beamforming rate even with a large number of IRS elements; 2) The effect of multi-user diversity is well captured, even with a small number of users.  

\begin{figure}[!t]
\includegraphics[width=\linewidth]{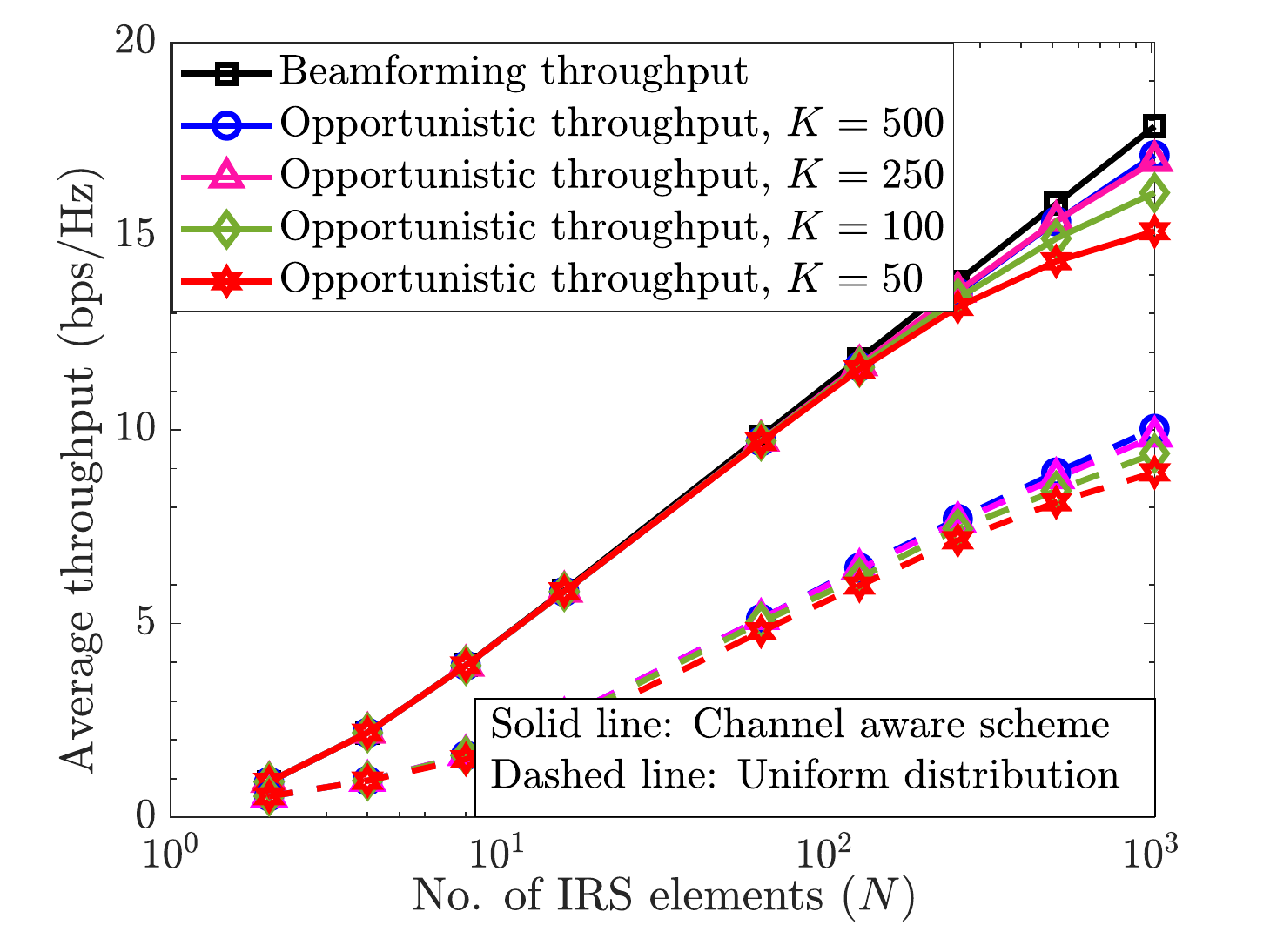}
\caption{Average throughput as a function of the number of IRS elements, for different number of users.}\label{fig6:rate_vs_IRS}
\vspace{-0.1cm}
\end{figure}

As the last experiment in this section, we consider the performance of IRS aided OC over wideband channels \textcolor{black}{as discussed in Sec.~\ref{sec.opp_schemes_wideband}}. \textcolor{black}{We fix the number of time domain channel taps to two different values:  $L=25$~\cite{Bjornson_spcup_2021} and $5$ with $\nu=1$, in \eqref{eq:pdp}, and perform communication through an OFDM system with $M=1024$ subcarriers and subcarrier spacing of $30$ kHz. As a result, the total system bandwidth is $30.72$ MHz which corresponds to a  total noise variance $\sigma^2 = -98.95 $~dBm at $300$~K. We choose a total power budget $P=24$~dBm at the BS and as a consequence, the nearest user experiences an average per subcarrier SNR of $11$~dB and the farthest user $-1.2$~dB, respectively.}
The channels are generated in an i.i.d.\ fashion across the users by setting the path loss coefficient to be equal for all users, such that the average SNR is $4.3$~dB, and the BS power is allotted divided across all the subcarriers. Before we look at the numerical performance of the schemes, we first ascertain the applicability of analytical rate scaling law in \eqref{eq:su-ofdm-rate-law} for the choice of $L=25$ at $\nu=1$ \textcolor{black}{in SU-OFDM systems}. To characterize the Gaussianity of the sum-term in the left hand side of \eqref{eq:sum_gaussian}, we compute its excess kurtosis, which measures how close a given distribution is to the Gaussian distribution \cite{Najim_book_2004}. The excess kurtosis, $\kappa$, of a random variable $X$ with mean $\mu$ is defined as 
\begin{equation}
\kappa = \frac{\mathbb{E}[|X-\mu|^4]}{\left(\mathbb{E}[|X-\mu|^2 \right)^2} - 3.
\end{equation}


\begin{table}[]
\centering
\caption{Excess kurtosis $\kappa$ as a function of the number of taps.}

\begin{tabular}{|c|c|c|c|c|c|c|c|c|}
\hline
$L$ & $1$ & $2$ & $5$ & $10$ & $20$ & $\mathbf{25}$ & $50$ & $100$ \\\hline
$\kappa$ & $5.96$ & $3.52$ & $1.5$ & $0.76$ & $0.35$ & $\mathbf{0.28}$ & $0.13$ & $0.07$ \\ \hline
\end{tabular}
\label{table:kurtosis}
\end{table}



\begin{figure*}[t]
	\centering
	\begin{subfigure}[t]{0.49\linewidth}
		\centering
		\includegraphics[width=\linewidth]{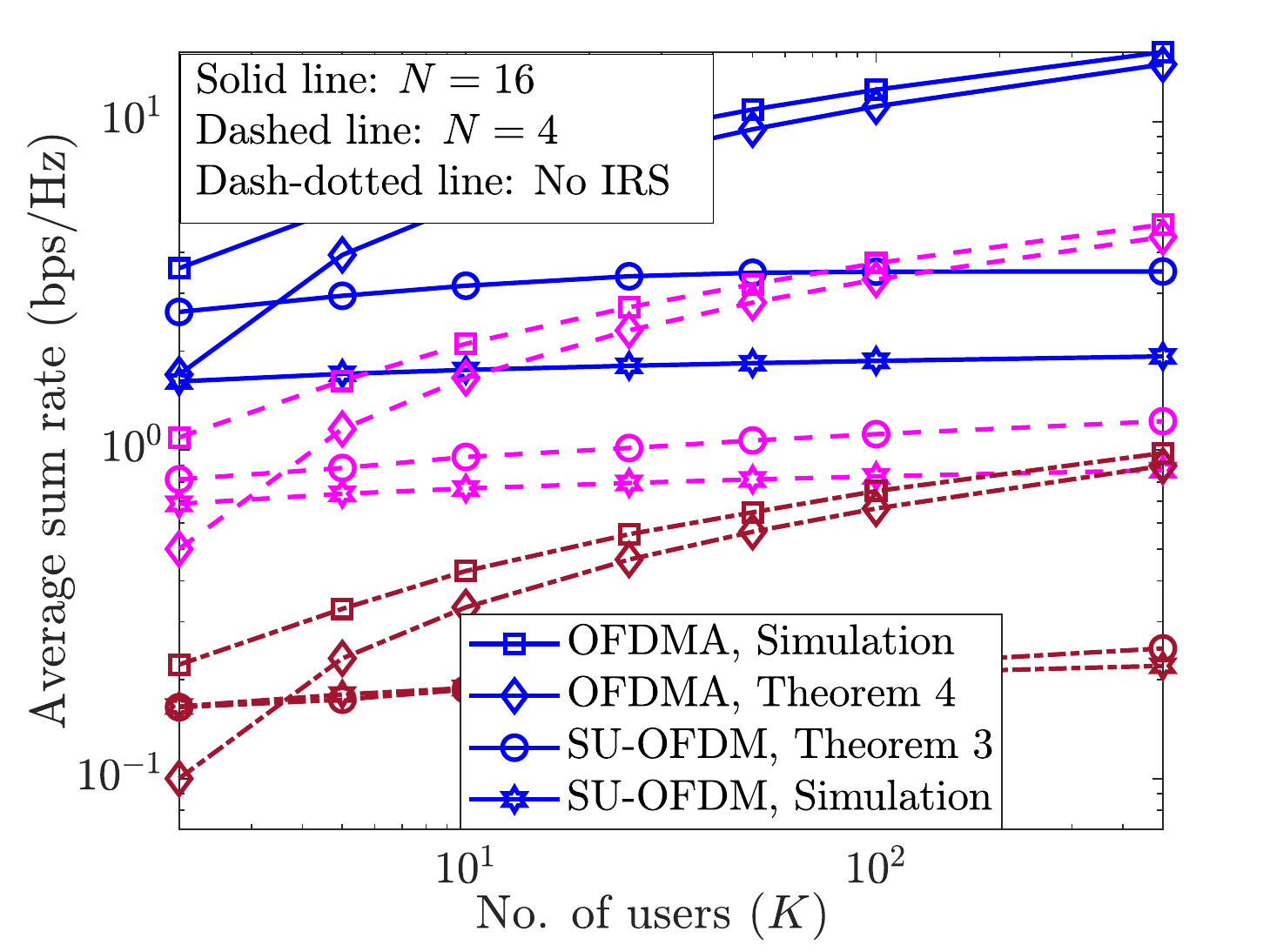}
		\caption{$L=25$.}
		\label{fig:ofdm-L-25}
	\end{subfigure}
	\begin{subfigure}[t]{0.49\linewidth}
		\centering
		\includegraphics[width=\linewidth]{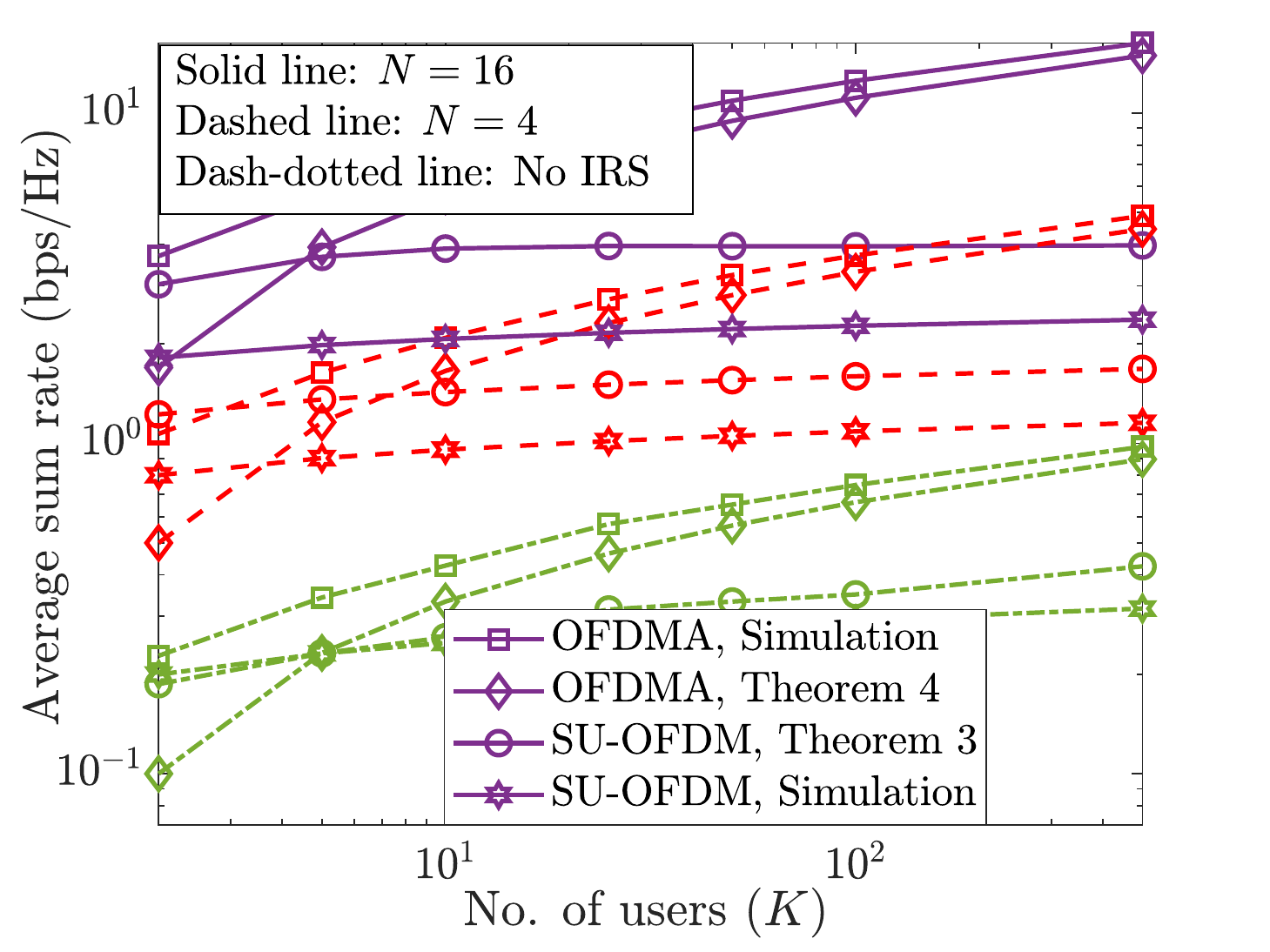}
		\caption{$L=5$.}
		\label{fig:ofdm-L-5}
	\end{subfigure}
	\caption{Average sum rate as a function of number of users, for an OFDM based communication system with $M=1024$ subcarriers.}
	\label{fig7:ofdm_opp}
\end{figure*}

For a Gaussian random variable, $\kappa = 0$. In \textcolor{black}{Table~\ref{table:kurtosis}, we list the excess kurtosis of the $L$-sum-term in \eqref{eq:sum_gaussian} as a function of $L$}. We see that, for $L=25$, we obtain an excess kurtosis of approximately $0.28$, \textcolor{black}{which is within $4\%$ of the distance between a Gaussian and exponential random variable (corresponding to $L=1$)} \textcolor{black}{and closer to the Gaussian random variable.} Thus, it is reasonable to consider that this sum-term is nearly  Gaussian distributed, making the scaling law in \eqref{eq:su-ofdm-rate-law} valid for~$L=25$. \textcolor{black}{Subsequently, we relax the requirement on large values for $L$, and study the validity of scaling law for smaller $L$ (particularly, at $L=5$.)}

\textcolor{black}{We present the OC performance of SU-OFDM and OFDMA for $L=25$ and $5$, respectively, in Figs.~\ref{fig:ofdm-L-25} and \ref{fig:ofdm-L-5}. Herein, we show the empirical and theoretical scaling of system throughputs under both the schemes. In particular, we have set the constant in the upperbound in~\eqref{eq:su-ofdm-rate-law} to $1$ for plotting the theoretical result.} \textcolor{black}{The throughput offered by the IRS assisted OFDMA is superior to that of SU-OFDM, in line with Remark~\ref{rem:OFDMvsOFDMA}. We also see that the performance of both OFDMA and SU-OFDM increase with $N$, in an IRS-assisted system exploiting multi-user diversity. However,  while the simulated and theoretical performance of OFDMA match well as $K$ increases, there is a gap between the two in the case of SU-OFDM. This is partly because the throughput analysis of SU-OFDM is an upper bound (see \eqref{eq:su-ofdm-jensen-parsevel}), and partly because the channels across the subcarriers become more disparate as the number of IRS elements increases, making the upper bound looser as $N$ increases. \textcolor{black}{We also observe that although the scaling law in Theorem~\ref{thm:su_ofdm_rate_scale} was derived assuming a large $L$, the expression captures the rate scaling performance even for moderate values of $L$ \textcolor{black}{such as $5$} (in Fig.~\ref{fig:ofdm-L-5}).} Nonetheless, the plot  shows that one can obtain a performance boost by deploying an IRS, even without using complex optimization algorithms, by merely obtaining multi-user diversity gains over randomly configured IRSs.}

\section{Conclusions}\label{sec.conclusion}
In this paper, we presented several opportunistic schemes in a single IRS aided setting for exploiting and enhancing the multi-user diversity gains, both in narrowband and wideband channels. The schemes completely avoid the need for CSI estimation and computationally expensive phase optimization, and require little or no communication from the BS to IRS. First, we saw that, in narrowband channels, a basic multi-user diversity scheme using a randomly configured IRS provides a performance boost over conventional systems as the number of users, $K$, gets large. In order to improve the rate of convergence of the opportunistic rate to the \textcolor{black}{optimal} rate (in terms of the number of users), we presented two alternative approaches \textcolor{black}{and analyzed their performances}: one where we obtained additional reflection diversity, and the other where we exploited the channel structure in IRS assisted systems. Both these schemes  improve the rate of convergence of the throughput from the OC schemes with the number of users. In particular, exploiting the channel structure allows us to significantly increase the number of IRS elements (\textcolor{black}{and also achieve the coherent beamforming throughput}) without requiring an exponentially larger number of users to achieve significant multi-user diversity gains. Finally, we considered IRS aided OC over a wideband channel in an OFDM system, and analyzed the performance of two different schemes, namely, SU-OFDM and OFDMA. Overall, IRS assisted OC schemes offer significant performance improvement over conventional schemes, while \textcolor{black}{incurring very low system overheads.} \textcolor{black}{Potential directions for future work could be to extend the setting to more general mmWave systems, multiuser MIMO-OFDM systems, and to develop OC schemes in an IRS-aided framework.}

{\appendices

\color{black}	
\section{Proof of the Theorem \ref{thm:mulitple_pilot_rate}}\label{app:mul_pilot_txn_rate_nb}
The proof of the theorem uses the following lemma on the extreme values of i.i.d.\ random variables.
\begin{lemma}\label{EVT_1} (\!\!\cite{Viswanath_TIT_2002}). 
Let $z_1,\ldots, z_K$ be i.i.d. random variables with a common cumulative distribution function (cdf) $F(\cdot)$ and probability density function (pdf) $f(\cdot)$ that  satisfy $F(z) < 1$ and is twice differentiable for all $z$. Let the corresponding hazard function, $\Omega(z) \triangleq \left[ {f(z)}\over {{1-F(z)}} \right]$ be such that
\begin{equation}\label{eq:evt_1}
\lim_{z \rightarrow \psi} \frac{1}{\Omega(z)} = c > 0,
\end{equation}
for $\psi \triangleq \sup\{z:F(z)<1\}$ and some constant $c$. Then, 
\textcolor{black}{$\max_{1 \le k \le K} z_k - l_K$} 
converges in distribution to a Gumbel random variable with cdf  $e^{(-e^{-x/c})}$, where $F(l_K)=1-\frac{1}{K} $.

\end{lemma}
In words, the lemma states that, asymptotically, the maximum of $K$ i.i.d. random variables grows like $l_K$. From Theorem~\ref{thm:mulitple_pilot_rate}, it is clear that, $|h_{k,q}|^2 \sim \exp\left(\frac{1}{\beta(N+1)}\right)$ for all $k \in [K]$ and $q \in [Q]$. Further, since the phase angles are independently generated at the IRS during each pilot transmission, we have that $ \grave{h} \triangleq 	|h_{k,q}|^2$ form a set of $QK$ i.i.d. exponential random variables with mean $\beta(N+1)$. Thus we have the pdf, $f_{\grave{h}}(h) = \frac{1}{\beta(N+1)} e^{-\frac{h}{\beta(N+1)}}$ for $h \geq 0$ and hence the cdf, $F_{\grave{h}}(h) = 1-e^{-\frac{h}{\beta(N+1)}}$ for $h \geq 0$ and $0$ otherwise. Also, here $\psi=\infty$. Thus, we have, 
\begin{equation}
\lim_{h \rightarrow \infty} \frac{1}{\Omega(h)}  = \lim_{h \rightarrow \infty} {\left(\frac{e^{-\frac{h}{\beta(N+1)}}}{\frac{1}{\beta(N+1)}e^{-\frac{h}{\beta(N+1)}}}\right)}^{-1} \!\!= \frac{1}{\beta(N+1)} > 0.
\end{equation} 
Hence, by virtue of the above lemma, we have $l_{QK} = F^{-1}\left(1 - \frac{1}{QK}\right)$ and solving, we get $l_{QK} = \beta(N+1)\ln(QK)$, and applying the lemma to \eqref{eq:mul_pilot_intial_rate}, we get the desired result.

\section{Proof of the Theorem \ref{thm:su_ofdm_rate_scale}}\label{app:su_ofdm_rate}
 We now use the fact that $L$ is large and invoke the following version of the central limit theorem (CLT)~\cite{Billingsley_book_1995}:
\begin{lemma}[Lyapounov's Central Limit Theorem]
Suppose that $X_1, X_2,\ldots, X_n$ form a sequence of independent random variables such that $\forall i \in [n]$, define $\mathbb{E}[X_i] \triangleq \mu_i$, $\sigma_i^2 \triangleq \mathbb{E}[|X_i|^2]$ and $s_n^2 \triangleq \sum_{i=1}^{n}\sigma_i^2$. If for some $\delta >0$,
\begin{equation}\label{eq:LCLT_condition}
\mathcal{L} \triangleq \lim_{n \rightarrow \infty} \sum_{i=1}^{n} \frac{1}{s_n^{2+\delta}}\mathbb{E}\left[|X_i-\mu_i|^{2+\delta}\right] = 0,
\end{equation} then
\begin{equation}
\frac{1}{s_n}\sum_{i=1}^{n} \left(X_i - \mu_i \right) \xrightarrow[n \rightarrow \infty]{d} \mathcal{CN}(0,1),
\end{equation}
where $\xrightarrow[n \rightarrow \infty]{d}$ stands for convergence in distribution.
\end{lemma}
To check whether the random variables $\left\{ |h_{k,l}|^2 \right\}_l$ satisfy \eqref{eq:LCLT_condition}, let $\delta=1$. Then $\mathbb{E}\left[|X_i-\mu_i|^{3}\right] = 2\left((N+1)a_l\right)^3$. Further, $s_n^2 = \sum_{i=1}^{n}a_l^2(N+1)^2$, hence $s_n^3 = \left(\|\mathbf{a}\|_2(N+1) \right)^3$. Thus, $\mathcal{L} = \lim_{n \rightarrow \infty} 2\frac{\|\mathbf{a}\|_3^3}{\|\mathbf{a}\|_2^3}$. We can lower and upper bound the ratio in the right hand side as follows: 
\begin{equation*}
  2\frac{n \min\limits_{i \in [n]} a_i^3}{n^{3/2}\max\limits_{i \in [n]} a_i^3} \leq  2\frac{\|\mathbf{a}\|_3^3}{\|\mathbf{a}\|_2^3} \leq 2\frac{n \max\limits_{i \in [n]} a_i^3}{n^{3/2}\min\limits_{i \in [n]} a_i^3}.
\end{equation*} 
Hence, when $n \longrightarrow \infty$ and $\|\mathbf{a}\| < \infty$, both lower and upper bounds go to zero and thus by the sandwich theorem, $\mathcal{L} = 0$. Thus, the given exponential random variables satisfy the condition in the lemma. Therefore, we~have, 
\begin{equation}
 \frac{1}{(N+1)\|\mathbf{a}\|_2} \sum_{l=1}^{L} \left(|h_{k,l}|^2-(N+1)a_l\right) \xrightarrow[L \rightarrow \infty]{d} \mathcal{CN}(0,1),
\end{equation} 
which implies
\begin{equation}\label{eq:sum_gaussian}
\sum_{l=1}^{L} |h_{k,l}|^2 \xrightarrow[L \rightarrow \infty]{d} (N+1)\left\{1 + \|\mathbf{a}\|_2 \bar{h}_k\right\},
\end{equation} where $\bar{h}_k \sim \mathcal{CN}(0,1)$. Thus, \eqref{eq:su-ofdm-jensen-parsevel} can be rewritten as, 
\begin{multline}\label{eq:su_ofdm_intial_evt}
	R^{(K)}_{\text{SU-OFDM}} \leq \log_2\left(1+ \frac{\beta P}{\sigma^2}(N+1) \right. \\ \left.\times \left[ 1+ \|\mathbf{a}\|_2 \left\{\max\limits_{1 \leq k \leq K} \bar{h}_k\right\}\right]  \right).
\end{multline} 
To characterize the extreme value of $K$ i.i.d. Gaussian random variables, we note that Lemma \ref{EVT_1} cannot be used as \eqref{eq:evt_1} is not satisfied by Gaussian random variables. Instead, we use another lemma to characterize the extreme values of i.i.d. Gaussian random variables from~\cite{Arnold_book_2007}.

\begin{lemma}[Von Mises' sufficient condition for weak convergence of extreme values]\label{EVT_2}
Let $X_1, X_2, \ldots, X_K$ be i.i.d. random variables and define $M_K \triangleq \max\{X_1, X_2, \ldots, X_K\}$. Let $F(x)$ be an absolutely continuous cdf with $f(x)$ being the corresponding pdf. Let the corresponding hazard function be $\Omega(x) \triangleq \left[ {f(x)}\over {{1-F(x)}} \right]$ and let $\psi = \sup\{ x: F(x)<1\}$. If
\begin{equation}\label{eq:evt2_condition}
\lim_{x \rightarrow \psi} \frac{d}{dx}\left(\frac{1}{\Omega(x)}\right) = 0,
\end{equation} then,
\begin{equation}
 M_K - l_K \xrightarrow[K \rightarrow \infty]{d} G,
\end{equation} where $G$ is a Gumbel random variable with cdf, $e^{(-e^{-x/c})}$ and $l_K$ is given by $F(l_K) = 1-\frac{1}{K}$ for some constant $c >0$.
\end{lemma} 
As before, the result shows that the extreme value of $K$ i.i.d. random variables satisfying \eqref{eq:evt2_condition} grows like $l_K$, asymptotically. In what follows, we check the applicability of the lemma to $K$ i.i.d. Gaussian random variables. Clearly, $\psi = \infty$. Let $\Phi(x)$ and $f(x)$ be the cdf and pdf of a standard complex normal random variable. We then have,
\begin{equation}\label{eq:evt2_check_gaussian}
\frac{\rm d}{{\rm d}x}\left(\frac{1-\Phi(x)}{f(x)} \right) = \frac{(1-\Phi(x))x}{f(x)} - 1 = \sqrt{2}x\frac{Q(x)}{f(x)} - 1,
\end{equation} where $Q(x)$ is the standard $Q$ function. But we know that,
\begin{equation}\label{eq:Q_bound}
\frac{1}{\sqrt{2\pi}x} e^{-\frac{x^2}{2}}\left(1-\frac{1}{x^2}  \right) \leq Q(x) \leq \frac{1}{\sqrt{2\pi}x} e^{-\frac{x^2}{2}}.
\end{equation} Using \eqref{eq:Q_bound} to lower and upper bound \eqref{eq:evt2_check_gaussian} and taking the limit $x$ to $\infty$ and applying the sandwich theorem, it is straightforward to show that \eqref{eq:evt2_condition} is satisfied by Gaussian random variables. Thus, we have that, $l_K = \Phi^{-1}\left(1-\frac{1}{K} \right)$. Substituting it in \eqref{eq:su_ofdm_intial_evt} yields the desired result in~\eqref{eq:su-ofdm-rate-law}.

}

%
%
\bibliographystyle{IEEEtran}
\bibliography{IEEEabrv,IRS_ref_short}



\section*{Supplementary material}

# Supplementary Material for Performance Analysis of Intelligent Reflecting Surface Assisted Opportunistic Communications

L. Yashvanth, *Student Member, IEEE*, Chandra R. Murthy, *Senior Member, IEEE*

## VII. REMARKS

### A. On the Comparison of IRS Channel Model Aware Multi-user Diversity scheme with SSBs in 5G NR

In the context of 5G NR, the proposed scheme can significantly improve the performance by deploying an IRS with negligible overheads. We note that the scheme using the synchronization signal blocks (SSBs) as as per NR specifications is *user centric*, wherein the BS sweeps through several beams across many time slots [50]. Thus, a given user measures the beam quality on every beam, and in the end, it reports to the BS the beam index of the beam which procured the best beam quality. This procedure focuses on obtaining the best beam for a given user and expends as many time resources as the number of beams which the BS sweeps. As a consequence, the time complexity scales linearly with the number of beams used for sweeping. Contrariwise, the proposed scheme is *IRS centric* where the random phase configuration selected at the IRS in a given time slot can be considered to be generation of a beam in a random (spatial) direction. The scheme then tries to locate the user who finds this random beam to yield the best beam quality. Thus, the proposed scheme does not need more than one time resource for scheduling a user, thus requiring lower overhead.

### B. On the Requirement of L for Theorem 3

To quantify how large $L$ should be in order for Theorem 3 to hold, we use the Berry–Esseen theorem [51] which describes the error bound in approximating the normalized sum of random variables to a normal random variable through the CLT. We can show that, the worst case absolute error, i.e., $\sup_{x\in\mathbb{R}} |F_n(x) - \Phi(x)| \lesssim C_1 \cdot a_1/\|\mathbf{a}\|_2$, where $F_n(x)$ and $\Phi(x)$ denote the cdf of the normalized sum-random variable and standard normal random variable, for some constant $C_1 > 0$. Thus, for a given $L$, this bound becomes tighter when the PDP is slowly decaying, with the power in the dominant tap being comparable to other taps (which occurs when $\nu$ is small), and in this case, the convergence in CLT occurs faster.

## APPENDIX C
## ON THE STATISTICS OF $h_{k,q}$

### A. Independent and Identical distribution of $h_{k,q}$

We recognize that $h_{k,q}$ (ignoring the direct path) can be decomposed at any time as

$$h_{k,q} = \sum_{i=1}^{N} \sqrt{\beta} h_{1,n} h_{2,k,n} e^{j\theta_{q,n}}. \tag{45}$$

The identicality of the distributions of $h_{k,q}$'s is clear due to $\beta_{r,k} \approx \beta_{d,k} = \beta$. In the sequel, we first prove that the channels are uncorrelated and then argue their independence under random and independent IRS phase configurations.

Since the analysis holds true for any user, we consider the channels at a single user with normalization of path losses, and subsequently extend the analysis across other users. Hence, we drop the subscript $k$ going forward, and let $f_n \triangleq h_{1,n}$, and $g_n \triangleq h_{2,k,n}$ for the sake of brevity. We then have from (45)

$$h_{k,q}/\beta \triangleq h_q = \sum_{i=1}^{N} f_n g_n e^{j\theta_{q,n}}. \tag{46}$$

We define the correlation coefficient as

$$\rho_{q,r} \triangleq \frac{\mathbb{E}[(h_q - \mathbb{E}[h_q])(h_r - \mathbb{E}[h_r])^*]}{\sqrt{\underbrace{\left(\mathbb{E}[|h_q|^2] - |\mathbb{E}[h_q]|^2\right)}_{\text{Var}(h_q)} \underbrace{\left(\mathbb{E}[|h_r|^2] - |\mathbb{E}[h_r]|^2\right)}_{\text{Var}(h_r)}}}. \tag{47}$$

We also decompose these complex random variables into their constituent real and imaginary parts as: $f_n = \Re(f_n) + j\Im(f_n)$, $g_n = \Re(g_n) + j\Im(g_n)$, and $e^{j\theta_{q,n}} = \cos(\theta_{q,n}) + j\sin(\theta_{q,n})$, so that $\Re(f_n), \Im(f_n), \Re(g_n), \Im(g_n) \stackrel{\text{i.i.d.}}{\sim} \mathcal{N}(0, \frac{1}{2})$, where $\Re(\cdot)$ and $\Im(\cdot)$ are the real and imaginary parts of a complex number. Then, the channel in (46) can be written as $h_q = X_q + jY_q$ where

$$X_q \triangleq \sum_{n=1}^{N} \{ \overbrace{(\Re(f_n)\Re(g_n) - \Im(f_n)\Im(g_n))}^{\triangleq a_n} \cos(\theta_{q,n}) \\ - \underbrace{(\Re(f_n)\Im(g_n) + \Im(f_n)\Re(g_n))}_{\triangleq b_n} \sin(\theta_{q,n}) \}, \tag{48}$$

$$Y_q \triangleq \sum_{n=1}^{N} \{ \overbrace{(\Re(f_n)\Re(g_n) - \Im(f_n)\Im(g_n))}^{= a_n} \sin(\theta_{q,n}) \\ + \underbrace{(\Re(f_n)\Im(g_n) + \Im(f_n)\Re(g_n))}_{= b_n} \cos(\theta_{q,n}) \}. \tag{49}$$

**Computation of** $\mathbb{E}[h_q]$: Clearly, $\mathbb{E}[h_q] = \mathbb{E}[X_q] + j\mathbb{E}[Y_q]$. Since $\theta_{q,n}$'s are sampled from uniform distribution in $[0, 2\pi]$, $\mathbb{E}[\cos(\theta_{q,n})] = \mathbb{E}[\sin(\theta_{q,n})] = 0 \ \forall q, n$. Since $\theta_{q,n}$ is independent of $f$ and $g$, $\mathbb{E}[X_q] = \mathbb{E}[Y_q] = 0$ by linearity of expectation over finite sums, and hence $\mathbb{E}[h_q] = \mathbb{E}[h_r] = 0$.

Hence, (47) becomes

$$\rho_{q,r} = \frac{\mathbb{E}[h_q h_r^*]}{\sqrt{\text{Var}(h_q)\text{Var}(h_r)}}. \quad (50)$$

**Computation of** $\mathbb{E}[h_q h_r^*]$: It is again clear that

$$\mathbb{E}[h_q h_r^*] = \mathbb{E}[(X_q + jY_q)(X_r - jY_r)]$$
$$= \mathbb{E}[X_q X_r] + \mathbb{E}[Y_q Y_r] + j\left(\mathbb{E}[X_r Y_q] - \mathbb{E}[X_q Y_r]\right). \quad (51)$$

Evaluating the first term in the RHS of (51), we have

$$\mathbb{E}[X_q X_r] = \mathbb{E}\left\{\left[\sum_{n=1}^N (a_n \cos(\theta_{q,n}) - b_n \sin(\theta_{q,n}))\right]\right.$$
$$\left.\times \left[\sum_{m=1}^N (a_m \cos(\theta_{r,m}) - b_m \sin(\theta_{r,m}))\right]\right\}. \quad (52)$$

This can be expanded as: $\mathbb{E}[X_q X_r] =$

$$\mathbb{E}\left\{\sum_{n=1}^N \sum_{m=1}^N a_n a_m \cos(\theta_{q,n})\cos(\theta_{r,m})\right.$$
$$-\sum_{n=1}^N \sum_{m=1}^N a_n b_m \cos(\theta_{q,n})\sin(\theta_{r,m})$$
$$-\sum_{n=1}^N \sum_{m=1}^N a_m b_n \cos(\theta_{r,m})\sin(\theta_{q,n})$$
$$\left.+\sum_{n=1}^N \sum_{m=1}^N b_m b_n \sin(\theta_{r,m})\sin(\theta_{q,n})\right\}. \quad (53)$$

Focusing on the first term in RHS of (53), we have by linearity of expectation over finite summations, and by independence of IRS phase angles and channel fading coefficients, $\forall n, m$,

$$\mathbb{E}\left\{\sum_{n=1}^N \sum_{m=1}^N a_n a_m \cos(\theta_{q,n})\cos(\theta_{r,m})\right\} =$$
$$\sum_{n=1}^N \sum_{m=1}^N \mathbb{E}[a_n a_m]\mathbb{E}[\cos(\theta_{q,n})\cos(\theta_{r,m})]. \quad (54)$$

But, from the properties of trigonometric functions we have

$$\cos(\theta_{q,n})\cos(\theta_{r,n}) = \frac{1}{2}\left(\cos(\theta_{q,n} + \theta_{r,n}) + \cos(\theta_{q,n} - \theta_{r,n})\right). \quad (55)$$

Also, we have the following

$$\mathbb{E}[\cos(\theta_{q,n} \pm \theta_{r,n})] \stackrel{(a)}{=} \frac{1}{4\pi^2} \int_0^{2\pi}\int_0^{2\pi} \cos(\omega \pm \phi)d\omega d\phi, \quad (56)$$

where (a) follows from the fact that IRS phase configurations are independent across IRS elements and pilots, and hence the joint distribution is the product of its marginals. By computing the above integral, we can show that $\forall n, m$, $\mathbb{E}[\cos(\theta_{q,n} \pm \theta_{r,n})] = 0$. Therefore, using (55) and (56), the expectation in (54) is 0, and hence is the first term in (53). We can similarly show that all the other three terms in (53) equal to 0, in turn making the first term in (51) equals to 0.

Subsequently, in the similar spirit, we can show that all the three other terms in (51) are also 0. Hence, $\mathbb{E}[h_q h_r^*] = 0$. It is also straightforward show that $\text{Var}(h_q) = \text{Var}(h_r) > 0$ (see Appendix C-B.) Hence, from (50), $\rho_{q,r} = 0$. Extending the argument across all users, we prove that $h_{k,q}$ form a set of uncorrelated random variables. Now, under the joint Gaussianity assumption [52] (see Appendix C-B also), these uncorrelated random variables are also independent. Finally, adding the direct path in the scenario does not change the result because direct paths seen across users are considered to be independent, and using the fact that functions of independent random variables are independent, the result follows. This completes the proof.

### B. On the distribution of $h_{k,q}$

From Appendix C-A, the cascaded channel ignoring path loss at a given user is given in (46). We also consider the real and imaginary parts as per (48),(49). Clearly, we see that $X_q$ and $Y_q$ are obtained as a sum over $N$ terms of random variables. Then for large $N$, by central limit theorem (CLT), we have that $X_q$ and $Y_q$ are normal random variables. We also know from Appendix C-A that they are zero-mean random variables. What remains is to compute their variances.

**Computation of** $\text{Var}(X_q)$: It is first clear that $\text{Var}(X_q) = \mathbb{E}[X_q^2]$. We now have

$$X_q^2 = \sum_{n=1}^N \sum_{m=1}^N [a_n \cos(\theta_{q,n}) - b_n \sin(\theta_{q,n})]$$
$$\times [a_m \cos(\theta_{q,m}) - b_m \sin(\theta_{q,m})]. \quad (57)$$

We can then expand (57) as follows:

$$X_q^2 = \sum_{\substack{n,m=1 \\ n=m}}^N (a_n \cos(\theta_{q,n}) - b_n \sin(\theta_{q,n}))^2$$
$$+ \sum_{\substack{n=1 \\ n\neq m}}^N \sum_{m=1}^N \{[a_n \cos(\theta_{q,n}) - b_n \sin(\theta_{q,n})]$$
$$\times [a_m \cos(\theta_{q,m}) - b_m \sin(\theta_{q,m})]\}. \quad (58)$$

We then simplify the first term in (58) as given below.

$$\sum_{\substack{n,m=1 \\ n=m}}^N \left(a_n^2 \cos^2(\theta_{q,n}) + b_n^2 \sin^2(\theta_{q,n})\right.$$
$$\left. -2a_n b_n \cos(\theta_{q,n})\sin(\theta_{q,n})\right). \quad (59)$$

It is also straightforward to compute the following values.
1) $\mathbb{E}[\cos(\theta_{q,n})] = \mathbb{E}[\sin(\theta_{q,n})] = 0$,
2) $\mathbb{E}[\cos^2(\theta_{q,n})] = \mathbb{E}\left[\frac{1+\cos(2\theta_{q,n})}{2}\right] = \frac{1}{2}$,
3) $\mathbb{E}[\sin^2(\theta_{q,n})] = \mathbb{E}\left[\frac{1-\cos(2\theta_{q,n})}{2}\right] = \frac{1}{2}$, and
4) $\mathbb{E}[a_n^2] = \mathbb{E}[b_n^2] = \frac{1}{2}$.

We can similarly compute the values for $m$−indexed variables, and the values remain the same as above. Subsequently, taking

expectations on both sides of (58), and invoking the linearity of expectations over finite sums on (58),(59), we can show that $\mathbb{E}[X_q^2] = \frac{N}{2}$. Deriving in the similar manner, we can show that $\text{Var}(Y_q) = \mathbb{E}[Y_q^2] = \frac{N}{2}$. Thus, by CLT, $X_q \sim \mathcal{N}(0, \frac{N}{2})$, and $Y_q \sim \mathcal{N}(0, \frac{N}{2})$. Since $X_q$ and $Y_q$ are i.i.d., it is evident that $h_q \sim \mathcal{CN}(0, N)$. Finally, adding the contributions from path loss and direct path, it follows that $h_{k,q} \sim \mathcal{CN}(0, \beta(N+1))$.



\end{document}